\documentclass[12pt,a4paper]{article}

\usepackage[utf8]{inputenc}
\usepackage[english]{babel}
\usepackage{amsthm}
\usepackage{amsmath}
\usepackage{amsfonts}
\usepackage{amssymb}
\usepackage{graphicx}
\usepackage{sgame}
\usepackage{color}
\usepackage{tikz}
\usetikzlibrary{angles,quotes}
\usepackage[spaces,hyphens]{url}
\usepackage{harvard}
\usepackage{enumerate}
\usepackage{hyperref}
\usepackage{subcaption}

\newcommand{\cG}{\mathcal{G}}

\newcommand{\E}{\mathbb{E}}
\newcommand{\xl}{{x_l}}
\newcommand{\xu}{{x_h}}

\newcommand{\R}{\mathbb{R}}

\newcommand{\res}{\,\rule[-4pt]{0.4pt}{11pt}\,{}}
\newcommand{\Gp}{G^{+}}
\newcommand{\Gm}{G^{-}}
\newcommand{\Gpm}{G^{\pm}}

\newcommand{\gm}{g^{-}}
\newcommand{\gpm}{g^{\pm}}
\renewcommand{\a}{x^0}
\renewcommand{\b}{x^1}

\DeclareMathOperator{\id}{id}
\DeclareMathOperator{\Tri}{Tri}

\newtheorem{theorem}{Theorem}[]
\newtheorem{assumption}{Assumption}[]
\newtheorem{lemma}[]{Lemma}
\newtheorem{proposition}[]{Proposition}
\newtheorem{cor}[]{Corollary}
\newtheorem{definition}[]{Definition}
\newtheorem{example}[]{Example}

\begin{document}

\title{The Texas Shoot-Out under Knightian Uncertainty}

  \author{Gerrit Bauch\thanks{Center for Mathematical Economics, Bielefeld University. Financial support through the German Research Foundation Grant Ri-1128-9-1, the German Academic Exchange Service and the Bielefeld Graduate School of Economics is gratefully acknowledged. For detailed comments and suggestions on the working paper version I would like to thank Roland Stauber. I am grateful for the hospitality of the Economics Department at the University of Arizona and helpful comments of their members as well as for fruitful discussions with participants of the Stony Brook conference on Game theory 2022 and the SAET 2022. Corresponding e-mail address \href{mailto:gerrit.bauch@uni-bielefeld.de}{gerrit.bauch@uni-bielefeld.de}} \and Frank Riedel\thanks{Center for Mathematical Economics, Bielefeld University and School of Economics, University of Johannesburg. Financial support through the German Research
Foundation Grant Ri-1128-9-1 is gratefully acknowledged.}}
\date{ \today}

\maketitle

\begin{abstract}\noindent
We investigate the allocation of a co-owned company to a single owner using the Texas Shoot-Out mechanism with private valuations. We identify Knightian Uncertainty about the peer's distribution as a reason for its deterrent effect of a premature dissolving. Modeling uncertainty by a distribution band around a reference distribution $F$, we derive the optimal price announcement for an ambiguity averse divider. The divider hedges against uncertainty for valuations close to the median of $F$, while extracting expected surplus for high and low valuations. The outcome of the mechanism is efficient for valuations around the median. A risk neutral co-owner prefers to be the chooser, even strictly so for any valuation under low levels of uncertainty and for extreme valuations under high levels of uncertainty. If valuations are believed to be close, less uncertainty is required for the mechanism to always be efficient and reduce premature dissolvements.
\end{abstract}

\noindent
{\footnotesize{\it Key words and phrases:} Knightian Uncertainty in Games, Texas Shoot-Out, Partnership Dissolution\\
{\it  JEL subject classification: C72, D74, D81, D82 } }

\section{Introduction}

Even the friendliest relationships can eventually turn sour. Co-owned companies are no exception, being at risk of feuding partners. In the event that co-owners are not able to cooperatively lead the company effectively anymore, it is reasonable to no longer sustain the partnership. There are two main kinds of solution concepts available to govern the dissolution of a partnership. On the one hand, \emph{external solutions} require a third party or stock market to serve as an outside option, e.g. a sale of the company to a prospective buyer or a liquidation in order to distribute asset shares can terminate a relationship and pay off the co-owners. However, those options are often not desirable since a liquidation goes along with a loss of jobs, pay-off of the business's debts and far reaching tax consequences while a new owner of the company may not be available on short notice. Furthermore, any external exit solution suffers from possible loss in market value of the company due to the publicly observed hostility within. Even more, private valuations of the co-owners are not taken into account who still might want to lead the company even without a partner. On the other hand, \emph{internal solutions} can mitigate the aforementioned drawbacks, while also taking into account the co-owners valuation for the company. Writing any such \emph{exit mechanism} into the buy-sell contract when founding a business allows an internal solution to be immediately available independent of whether or not an external outside option is at hand.

For a two party co-owned company, a frequently used exit mechanism is called the \emph{Texas Shoot-Out}. Many consulting platforms\footnote{See e.g., {ClaytonCapitalPartners} or {ExitPlanningSolutions} at \url{https://claytoncapitalpartners.com/navigator/issue66-biz_continuity_part3.html} and \url{https://www.exitplanning.com/blog/texas-shootout}.}
recommend it for its simplicity and effect as a deterrent to a premature dissolving. The Texas Shoot-Out works as follows: Any partner has the right to initiate the exit mechanism at any time she desires, becoming the so-called ''divider''. The divider commits to a price $p$ for the sole ownership of the company. Her partner, called the ''chooser'' now has exactly two options. Either, she buys the ownership and interest of the divider at price $p$ or sells hers at price $p$. In the end, only one partner remains as the sole owner of the business, having compensated the former co-owner. Although this mechanism is independent of actual shares of the co-owners, it is typically recommended for equally sized shareholders.

While simple, the Texas Shoot-Out is notorious for its deterrent effect on feuding partners: Neither partner knows in which role they will eventually find themselves, what price they might face in case the co-owner initiated the exit mechanism or what choice to expect of her if they trigger the mechanism themselves. The so created uncertainty leads to a very purposeful execution of the mechanism and thus encourages conciliation in times of minor conflicts while at the same time offering a promptly available tool in order to solve a dispute if the partnership reached a dead end.

The focus of this article is to precisely characterize the deterrent component of the Texas Shoot-Out. Since there is no objective information about a peer's valuation, it is natural to assume that co-owners face imprecise probabilistic information about the valuation of their peers. We model this \emph{Knightian Uncertainty} by means of a set of priors about the distribution function of the co-owner's valuation. We analyze such a prior set in the form of \emph{distribution bands} since they allow for more flexibility than uncertainty about the valuation alone, being able to incorporate bounds for main characteristics of the distribution, such as its expectation or variance. More precisely, a co-owner is willing to entertain the belief that the partner's valuation is close to a reference distribution $F$, but prefers a robust approach to the uncertainty about the exact distribution. The proposed framework thus models the Texas Shoot-Out in between the Bayesian (single prior) setting and the full uncertainty (all priors) case explaining what happens for intermediate levels of Knightian Uncertainty about the co-owner's distribution. Requiring mainly strict quasiconcavity of the induced payoff function, our main finding is the astonishing link that connects those two. The optimal price announcement of a divider is an increasing function in her own valuation with two kinks left and right from the median of $F$. In between, she announces half her valuation, thus hedging herself against uncertainty. For low (resp. high) valuations she plays as if facing the Bayesian setting induced by the distribution function being stochastically dominated by (resp. stochastically dominating) all other distributions in her prior set. Increasing the level of uncertainty about the co-owner's distribution also increases the interval around the median that corresponds to a full-hedging strategy. This completely pins down the continuous transformation from the Bayesian case to the one of full uncertainty: Hedging behavior starts at the median of $F$ and spreads to more and more extreme low and high valuations continuously with an increase in uncertainty. As long as uncertainty is not too high, dividers with very low and high valuations make strategic price announcements. E.g., for low valuations, the divider states a price exceeding half her valuation, expecting the chooser to still accept the offer and thus generating a revenue for herself.

Despite allowing for a robust and thus more flexible and realistic approach to partnership dissolvements, introducing Knightian Uncertainty to the Texas Shoot-Out elicits not only desirable properties but also explains its notorious deterrent effect for comparative statics analyses.\\
Efficiency of the mechanism, in the sense of giving the company to the co-owner with highest valuation, is increased by a boost to uncertainty: More valuation profiles will lead to an efficient allocation compared to the classical Bayesian setting. Especially, a precisely determined range of valuations of a divider close to the median of $F$ can be guaranteed to produce efficient allocations whereas the Bayesian setting could only assure this for a single valuation. This range is the larger the higher the level of uncertainty about the chooser's valuation.\\
Incentives for a profitable selfish exit are diminished the higher the uncertainty no matter the role a co-owner finds herself in, especially for the initiator (divider): Fixing a valuation for the company, the interim worst-case expected utility for a chooser is always higher than for a divider. This preference is always strict for low levels of uncertainty and stays strict only for very low and high valuations if uncertainty is increased until it is too high and both are indifferent. This finding explains why the Texas Shoot-Out is preventing a premature selfish end of the partnership: Only co-owners with an extreme low or high valuation expect a revenue exceeding half their valuation from initiating the exit mechanism. However, they thus become the divider and suffer from a lower interim expected payoff than the other co-owner. Hence, the more uncertainty, the fewer types will execute the Texas Shoot-Out without good reason to do so and expect to get a utility close to half their valuation. At the same time, the Texas Shoot-Out does offer a fair way out of a dead end since every co-owner expects to get no less than half her valuation.\\
Efficiency and deterrence of the Texas Shoot-Out do not require complete ignorance of the partner's valuation. It seems likely that over the course of the partnership, co-owners share a similar (objective) belief about the company's value. In other words, a co-owner suspects the peer's valuation to be close or correlated to her own. Surprisingly, this additional intelligence will decrease a co-owner's incentive for a selfish exit even for moderate amounts of uncertainty and would lead again to more efficiency and conservative price announcements if the Texas Shoot-Out was triggered.

Our contribution contrasts two standpoints on the Texas Shoot-Out in the literature by connecting them. In the well-known case of a single distribution $F$ (\cite{mcafee92}) the optimal price announcement describes a strictly increasing function above the line $\tfrac{x}{2}$ that touches it exactly once - at the median of $F$. Though the chooser can perfectly identify the divider's type in equilibrium, this strategic price announcement can lead to an inefficient outcome in which the co-owner with lower expectation obtains the company. In the maxmin equilibrium (\cite{van2020}), a co-owner prepares for the worst-case the co-owner can inflict on herself by always offering exactly half her valuation. The allocation will thus always be efficient at the cost of extreme behavioral assumptions.

The Texas Shoot-Out can be interpreted as variant of a \emph{cake-cutting mechanism} (see \cite{moulin2019fair} for a survey). Typically stated for divisible objects these mechanisms describe discrete or continuous procedures for proportional or envy-free allocations (\cite{brams1995envy}, \cite{brams1997moving}). Principally, any cake-cutting mechanism can be extended to settings with indivisible objects by introducing transfer payments, as done in the Texas Shoot-Out, or having a selling third party with sufficient information (\cite{glazer1989efficient}). The question which agent is cutting or choosing is significant with regard to an efficient allocation as pointed out by \cite{de2008efficient}. For the Texas Shoot-Out specifically, \cite{brooks2010trigger} analyze efficiencies and limitations under the assumption that the roles are randomly assigned or delegated to the more informed party (\cite{landeo2013shotgun}). For publicly known shares and private valuations of an indivisible good, \cite{cramton1987dissolving} derive the class of incentive compatible and individual rational mechanisms that allocate the sole ownership of an object to a shareholder which reconciles the negative result of \cite{myerson1983efficient} in which for a sole owner the object cannot be efficiently allocated.


The remainder of the article is organized as follows. Section \ref{Section: Bayes and maxmin} gives a formal introduction to the Texas Shoot-Out and briefly summarizes the results of the stochastic setting and the one of maxmin strategies, known from the literature. Section \ref{Section: Knight} adds Knightian Uncertainty to the setup. The main result - the optimal price announcement under uncertainty - is stated in Section \ref{Section: Optimal Price Announcement} and reveals the link between the for price announcements from the previous section. In addition to the efficiency of the allocation (Section \ref{Section: Efficiency}), interim expected utility is derived and compared between the co-owners (Section \ref{Section: interim utility}), explaining the deterrent effect of a premature selfish dissolution of the Texas Shoot-Out. Section \ref{Section: correlation} argues that only small amounts of uncertainty are necessary for the Texas Shoot-Out to prove its qualities if correlation of valuations is adopted. How our formal results are to be interpreted in and apply  more general settings is discussed in Section \ref{Section: Uncertainty only about valuations}. Finally, Section \ref{Section: Conclusion} wraps up our findings.

\section{Bayesian and Maxmin Equilibrium}\label{Section: Bayes and maxmin}

Consider two\footnote{We consider only the simplest two-player case. More general procedures for proportional or envy-free cake-cutting for $n$ players involve many more steps or 'trimming techniques', see \cite{brams1995envy}. We do not know of any of these $n$-player mechanisms finding actual use in company contracts for $n \geq 3$.} equal owners of a company who have come to the point where they want to dissolve their partnership. Both owners have a private valuation $x_D,x_C$ in a compact interval $X:= [\xl,\xu] \subset \R$ for being the sole owner of the company. In the \emph{Texas Shoot-Out}, the first player (''divider'') announces a price\footnote{We explicitly allow for negative price bids, of which we think as compensation, which might occur if a valuation for the company is negative. That way, the Texas Shoot-Out does not only apply to allocation decisions of desired objects, but also to undesirable objects, such as debt.} $p \in \R$ that she is willing to offer for obtaining or forsaking the sole-ownership. In the next step, the second player (''chooser'') either pays the divider $p$ and becomes the sole owner of the company or sells her shares at price $p$ for the divider to obtain the company. This procedure is the simplest cake-cutting mechanism with compensation payments: The Divider cuts, the chooser selects her piece.

The chooser has all her payoff-relevant information at hand and optimizes by simply 'picking the larger piece of the cake' from her perspective: Facing an offered price $p$ knowing her valuation $x_C$, she clearly takes the offer and sells her shares if the offered price exceeds her private value of the company minus the required payment, i.e. $p > x_C -p$, and buys the divider's shares of the company to become the sole owner if $p < x_C -p$. In the following, we stick to this behavior which must be part of any Bayesian equilibrium on-equilibrium-path\footnote{To rule out counter-intuitive off-equilibrium-path behavior and thus making this behavior unique, a refinement for extensive form games featuring sequential rationality can be applied, e.g., weak perfect Bayesian equilibrium or sequential equilibrium \cite{sequential}.} and is free of any assumptions on the distribution of valuations or the played strategy of the divider: The chooser surely maximizes her expected payoff by doing so for any of her valuations and any price $p$.
In the case of indifference, i.e. $p = x_C / 2$, we impose the tie-breaking action 'sell' as this is more in line with writing out the expected payoff of the divider by means of (right-continuous) distribution functions. Note that the so-defined strategy of is always a best reply for the chooser.

The divider's ideal action is more involved and will take up most of the upcoming analysis. She faces a trade-off between bidding a price either high enough to get the company while not going into debt or low enough to get payed while still getting something out of the deal. Let us assume that the divider has a belief over the chooser's valuation  given by the distribution function $F$ on $X$ with a strictly positive and differentiable density function $f$. Denote by $\mu_F$ its (unique) median. Furthermore, we allow the divider to entertain a utility function which is twice continuously differentiable, concave and strictly increasing with $u'>0$. Taking into account the aforementioned response of chooser\footnote{Recall that, by right-continuity of $F$, we implicitly assume that the chooser sells her shares if she is indifferent between accepting and rejecting the offer, i.e. if $2p=x_C$. This can be taken as a convention and is in fact only a null event if $F$ is atom-less. Allowing for point masses, this tie breaking rule determines whether or not the supremum of the payoff function can be attained, c.f. Section \ref{Section: Uncertainty only about valuations}}, the divider's (interim) expected payoff given her valuation $x_D$ is
\begin{equation*}
    \pi_F(p\mid x_D) = u(x_D-p) \cdot F(2p)+ u(p) \cdot (1- F(2p)).
\end{equation*}

\paragraph{Maxmin and Hedging}
Let us first consider the maxmin strategy that was recently discussed by \cite{van2020}. It supports the idea that the opponent will always choose the action hurting one the most - not taking into account their own losses. Note that the divider can remove any uncertainty about the payoff by simply bidding half of her own valuation, i.e. $\bar{p}=\frac{x_D}{2}$, guaranteeing a payoff of 
$$\pi_F( \bar{p} \mid x_D)= u \left(\frac{x_D}{2}\right) \cdot  F(x_D)+ u\left(\frac{x_D}{2}\right) \cdot (1-F(x_D))= u\left(\frac{x_D}{2}\right).$$
The choice $\frac{x_D}{2}$ thus hedges the divider's uncertainty completely and we call this strategy \emph{full-hedging}.

If the divider offered a price $p \neq \tfrac{x_D}{2}$ and thus (exactly) one of the two possible outcomes $u(x_D-p), u(p)$ is strictly smaller than $u(\tfrac{x_D}{2})$, the chooser would hurt the divider the most by choosing the action leading to the divider getting the lower outcome. Thus, the following statement holds.

\begin{proposition}[\cite{van2020}]
In maxmin equilibrium, the divider bids $\bar{p}=\frac{x_D}{2}$ and the chooser sells her shares if
$\bar{p} > x_C /2 $ and buys the divider' shares if $\bar{p} < x_C/2$.
\end{proposition}

The outcome of the maxmin equilibrium is efficient in the sense that the player with the highest valuation obtains the company.

\paragraph{Bayesian Equilibrium}
The Texas Shoot-Out can be modeled as a Bayesian Game by introducing a (not necessarily common) joint CDF(s) on the two agent's valuations. For ease of exposition, we stick to an i.i.d. distribution of valuations\footnote{More generally, we allow for correlation in Section \ref{Section: correlation}.} given by a CDF $F$ on $X$ in the following. Assume that $F$ is atom-less and admits a strictly positive and differentiable density function $f$. Denote by $\mu_F$ its (unique) median.
In Bayes-Nash equilibrium, the divider seeks to maximize her expected payoff given the aforementioned strategy of the chooser by picking a suitable price $p$ for each of her valuations $x_D$. She faces a trade-off in doing so: The higher the price, the more likely the chooser will sell her shares and the divider ends up paying the exorbitant price herself. The lower the price, the more likely the chooser will take the company and pays her a knocked-down price. Thus, the optimal price announcement will be a moderate and deliberate decision. Our first result thus is an immediate boundary for the divider's price decision.

\begin{lemma}\label{Lemma: Priceinterval Bayes}
The divider will only announce prices fulfilling $2p \in [\xl,\xu]$.
\begin{proof}
All of the proofs are delegated to the appendix.
\end{proof}
\end{lemma}

For the uniqueness of an optimal price announcement we impose strict quasiconcavity on the payoff function. For sake of completeness, we include its definition and relevant properties in Appendix \ref{Appendix: Quasiconcavity}.

\begin{assumption}\label{AssF}
For a considered valuation $x_D$, let the divider's expected payoff function $\pi_F(p \mid x_D)$ be strictly quasiconcave in $p$ for $2p \in [\xl,\xu]$.
\end{assumption}

In the literature, it is common to impose monotone hazard rate conditions on the prior. We show that these conditions are sufficient for strict quasiconcavity of the resulting payoff function, eliciting quasiconcavity as the underlying driver of well-known results.

\begin{lemma}\label{Lemma: Suff McAfee}
Assumption \ref{AssF} is satisfied for all valuations $x_D$ if the \emph{standard hazard rate conditions} (SHRCs, \cite{mcafee92}) hold, i.e.
\begin{align}
    \frac{\partial}{\partial x} \left( x + \frac{F(x)}{f(x)} \right) \geq 0 \quad \text{and} \quad \frac{\partial}{\partial x} \left( x - \frac{1-F(x)}{f(x)} \right) \geq 0.
\end{align}
\end{lemma}

Put simply, a strictly quasiconcave function $f$ has a unique point $m_f$ such that $f$ is strictly increasing before and strictly decreasing after $m_f$. It thus should not come as a surprise that it is related to the price that maximizes expected payoff and can be identified as the solution to the resp. FOC under smoothness assumptions on $F$. However, in the interplay with $F$, the optimal price reveals intuitive qualitative properties discovered by \cite{mcafee92} for a CDF fulfilling the SHRC and now stated under the assumption of strict quasiconcavity. Especially, there will be strategic choice to announce a price above what the divider thinks her shares are worth when she has a low valuation with the goal of getting paid more in the expected case of the chooser buying ones shares (and vice versa for high valuations).

\begin{proposition}\label{Proposition: McAfee}
The optimal price announcement of the divider is uniquely determined, given by $m(x_D) := m_{\pi_F(.\mid x_D)}$. Furthermore, $m(x_D)$ is strictly increasing in $x_D$ and fulfills
\begin{equation*}
    \begin{cases}
    x_D &< 2m(x_D) < \mu_F\\
    x_D &= 2m(x_D) = \mu_F\\
    x_D &> 2m(x_D) > \mu_F
    \end{cases}
    \quad \text{ if and only if } \quad 
    \begin{cases}
    x_D &< \mu_F\\
    x_D &= \mu_F\\
    x_D &> \mu_F
    \end{cases},
\end{equation*}
where $\mu_F$ is the median of $F$. Together with the chooser's strategy, this constitutes a Bayesian Nash equilibrium.

Furthermore, for any valuation $x$, interim expected utility is strictly larger for the chooser.
\end{proposition}

We now illustrate the implications of Proposition \ref{Proposition: McAfee} using the uniform distribution on the unit interval, see Figure \ref{Example: no uncertainty}. This example will be extended to allow for Knightian Uncertainty in the upcoming section.

\begin{example}\label{Example: no uncertainty}
Let us consider the particularly transparent case of the uniform distribution $f(x)=1$ for $x \in [0,1]$ and $u = \id$. By Lemma \ref{Lemma: Suff McAfee}, Assumption \ref{AssF} is satisfied. For $0\leq p \leq 1/2$ we have
$$\pi_F(p \mid x_D)= (x_D-p) 2p+p(1-2p)= 2x_D \cdot p+p-4p^2.$$
The optimal bid is given by the first order condition
\begin{align*}
 0&=2 x_D +1 - 8 p\\
\iff p &= x_D / 4 + 1 / 8.
\end{align*}
\end{example}

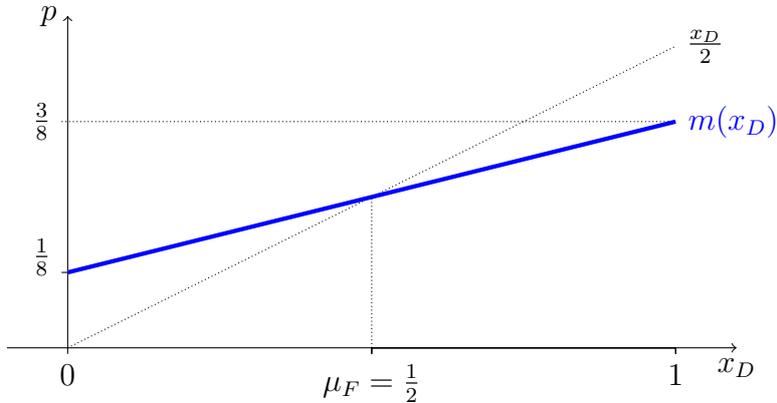
\begin{figure}[ht]
\begin{center}
\begin{tikzpicture}[scale=8]
\draw[<->] (0,1/2+0.05)node[left]{$p$}--(0,0)--(1,0)--(1,-0.01)node[below]{$1$}--(1,0)--(1/2,0)--(1/2,-0.01)node[below]{$\mu_F = \tfrac{1}{2}$}--(1/2,0)--(1+1/10,0)node[below]{$x_D$};
\draw (-1/10,0)--(0,0);
\draw (0,0)--(0,-1/100)node[below]{$0$};


\draw (0,1/8)--(-1/100,1/8)node[left,yshift=0.5em]{$\tfrac{1}{8}$};


\draw[densely dotted](0,0)--(1,1/2)node[right]{$\tfrac{x_D}{2}$};
\draw[densely dotted] (-0.01,3/8)node[left]{$\tfrac{3}{8}$}--(1,3/8);
\draw[densely dotted] (1/2,0)--(1/2,1/4);

\draw[scale=1, domain=0:1, smooth, variable=\x, blue,ultra thick]  plot ({\x}, {1/4 * \x + 1/8});
\draw (1,3/8)node[ultra thick, right,blue]{$m(x_D)$};
\end{tikzpicture}
\caption{Optimal price announcement
 for $F \sim \mathcal{U}([0,1]), u = \id$.}
\label{Figure: uniform price announcement stochastic}
\end{center}
\end{figure}

Note that the outcome of the Bayes-Nash equilibrium is not always efficient, in contrast to the maxmin equilibrium outcome. There is an incentive for dividers with low valuation to bid a relatively high price, because the chooser is going to accept with a certain probability. 
It can thus happen that the divider obtains the company although she has a lower valuation than the chooser in the case of low valuations and vice versa in case of both players having a high valuations. Section \ref{Section: Efficiency} will elaborate more on the issue of efficiency.

\section{Equilibrium under Knightian Uncertainty}\label{Section: Knight}

We now turn towards the case of \emph{Knightian Uncertainty}, i.e. imprecise probabilistic information, about the distribution of an agent's valuation, where a set of beliefs is deemed possible. More precisely, the divider considers a whole family of distributions $G \in \cG$ for the chooser's valuation, bounded by two distribution functions $G_0, G_1$ where $G_0$ stochastically dominates\footnote{The notion of stochastic dominance is as usual: $G$ stochastically dominates $G'$ if $G(x) \leq G'(x)$ for all $x$. Equivalently, $\int u(x) \, G(\mathrm{d}x) \geq \int u(x) \, G'(\mathrm{d}x)$ for all functions $u$ with $u' \geq 0$, see \cite{levy92} or \cite{rothschild1970}.} $G_1$:
\begin{equation*}
    \cG = \left\{\text{$G$ is a CDF on $X$ satisfying } G_0(x) \leq G(x) \leq G_1(x) \, \forall x \right\}.
\end{equation*}
Put differently, the probability that the chooser's valuation is below a value $x$ is assumed to be at least $G_0(x)$ but no more than $G_1(x)$. Such a set is usually referred to as a \emph{distribution band} (c.f. \cite{basu1995robust} and \cite{basu1994variations}) in robust Bayesian analysis. The distributions $G_0$ and $G_1$ are called \emph{lower} resp. \emph{upper bound} of $\cG$.

Uncertainty with distribution bands arise naturally when the divider is willing to entertain a belief $F$ about the chooser's valuation, yet prefers a robust approach to account for the uncertainty about the exact distribution. As a benchmark example, she might consider the family of distributions $G \in \cG$ that are \emph{$\epsilon$-shifts} of $F$ for any $\epsilon \geq 0$. Precisely, we define
\begin{equation*}
    \cG(F,\epsilon) := \left\{ G \text{ is a CDF on $X$ with } F(x-\epsilon) \leq G(x) \leq F(x+\epsilon) \, \forall x \right\}.
\end{equation*}

Varying $\epsilon$ allows for changes in the probability weights of each type as well as shifts of the entire distribution to higher and lower valuations. The size of $\epsilon$ can be small, even non-existent for $\epsilon = 0$, or large until $\cG(F,\epsilon)$ contains all distribution functions on $X$, representing \emph{full uncertainty}.

The defined distribution band $\cG(F,\epsilon)$ has nice topological properties with respect to the \emph{interval topology} ( c.f. \cite{nendel2020note}) in which it is closed and Dedekind super complete. Especially, $\cG(F,\epsilon)$ always admits a stochastically dominant resp. dominated representative $G_0$ resp. $G_1$. Furthermore and in contrast to the Prohorov metric\footnote{Also called Lévy-Prohorov metric, c.f. \cite{billingsley2013convergence}, defined by $d(F,G) := \inf\{ \eta>0 : F(x-\eta)-\eta \leq G(x) \leq F(x+\eta)+\eta \mbox{ for all $x \in [0,1]$}\}
$} the distribution band $\cG(F,\epsilon)$ is compatible with the SHRCs from Lemma \ref{Lemma: Suff McAfee} on $F$ as we will see in Section \ref{Section: Optimal Price Announcement}.
We illustrate $\cG(F,\epsilon)$ for the case of a uniform distribution in Figure \ref{Figure: epsilon shift distance}.

Concerning preferences in the presence of Knightian Uncertainty, we assume both agents to be uncertainty-averse in the sense of Gilboa and Schmeidler \cite{gilboa1989}. While the chooser can again decide only based on her own interim valuation and the announced price as before, the divider faces her interim worst-case expected utility\footnote{Considering the worst-case expected utility is a strong form of ambiguity aversion. One might be intrigued to apply weaker notions such as the smooth model \cite{smooth}. However, as distribution bands do not entail any information about the likelihood of one of its CDFs and their CDFs can be highly discontinuous, there's no natural parametrization or computationally accessible choice for a subjective probability over $\cG$. Maximizing the worst-case expected utility on the other hand can readily be defined without additional parameters. We thus find it reasonable to assume Gilboa-Schmeidler preferences, especially when considering $\cG(F,\epsilon)$, having an ambiguity averse divider in mind who chooses a robust approach to uncertainty by not making any assumptions on distributions functions close to a designated CDF $F$.} in the class $\cG$, given by
$$ \pi( p \mid x_D ) = \min_{G \in \cG} \pi_G(p\mid x_D).$$

By strict monotonicity of $u$, we have $u(x_D-p)>u(p)$ if and only if  $x_D > 2p$. Thus, if the divider announces a price $x_D > 2p$, her worst-case belief is $G(2p) = G_0(2p)$. Vice versa, her worst-case belief for $x_D < 2p$ is $G(2p) = G_1(2p)$. Consequently, only $G_0$ and $G_1$ are relevant for evaluating the worst-case for the divider, i.e.
\begin{align*}
\pi( p \mid x_D ) &= \min_{G \in \{ G_0,G_1 \}}  u(x_D-p)G(2p) + u(p) (1-G(2p))\\
    &= \begin{cases}
    u(x_D-p)G_0(2p) + u(p)(1-G_0(2p)) &, 2p<x_D,\\
     u(x_D-p)G_1(2p) + u(p)(1-G_1(2p)) &, x_D \leq 2p
     \end{cases}.
\end{align*}

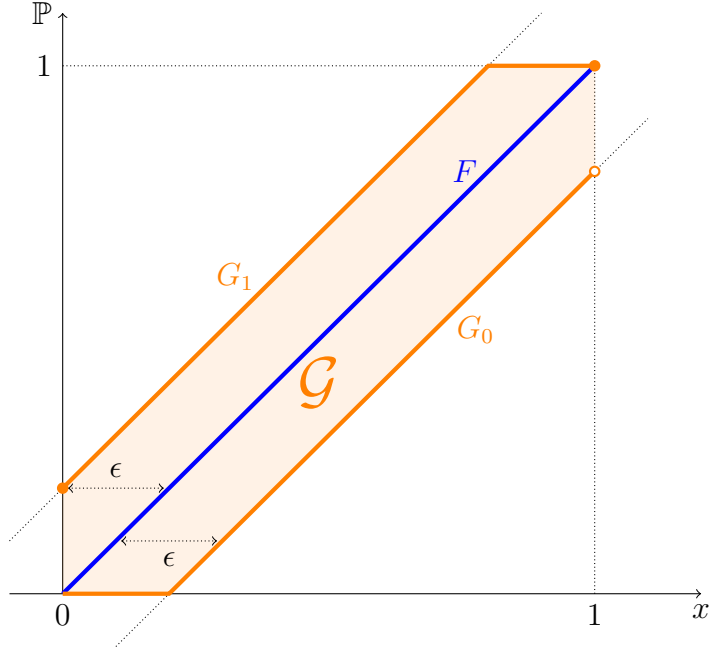
\begin{figure}[ht]
\begin{center}
\begin{tikzpicture}[scale=7]
\draw[->] (-0.1,0)--(1,0)node[below]{$1$}--(1.2,0)node[below]{$x$};
\draw[->] (0,0)node[below]{$0$}--(0,1)node[left]{$1$}--(0,1.1)node[left]{$\mathbb{P}$};
\draw[densely dotted] (0,1)--(1,1);
\draw[densely dotted] (1,0)--(1,1);

\draw[densely dotted](-0.1,0.1)--(0.9,1.1);
\draw[densely dotted](0.1,-0.1)--(1.1,0.9);

\filldraw[orange,opacity=0.1](0,0)--(0.2,0)--(1,0.8)--(1,1)--(0.8,1)--(0,0.2)--cycle;

\draw[densely dotted,<->](0.11,0.1)--(0.2,0.1)node[below]{$\epsilon$}--(0.29,0.1);
\draw[densely dotted,<->](0.19,0.2)--(0.1,0.2)node[above]{$\epsilon$}--(0.01,0.2);

\draw[ultra thick, blue] (0,0)--(1,1);
\draw[blue] (0.8,0.8)node[left]{$F$};
\draw[ultra thick, orange] (0,0)--(1/5,0)--(1,4/5);
\draw[ultra thick, orange] (0,1/5)--(4/5,1)--(1,1);
\draw [orange,thick,fill=orange] (0,1/5) circle(0.25pt);
\draw[orange,thick,fill=orange] (1,1) circle(0.25pt);
\draw[orange,thick,fill=white] (1,4/5) circle(0.25pt);
\draw[orange] (0.72,0.5)node[right]{$G_0$};
\draw[orange] (0.38,0.6)node[left]{$G_1$};
\draw[orange] (0.4,0.4)node[right,scale=2]{$\cG$};

\end{tikzpicture}
\end{center}
\caption{$\cG(F,\epsilon)$ for $F \sim \mathcal{U}([0,1])$ and $\epsilon = 1/5$.}
\label{Figure: epsilon shift distance}
\end{figure}

\phantom{a}\\
In the Texas Shoot-Out, the maxmin setting and the one of \emph{full uncertainty}, i.e. $\cG$ contains all CDFs on $X$, are linked and amount to the same behavior: For any announced price  $p$, there exists a valuation $x_C^*$ of the chooser that implements the worst-case possible action of the chooser for the divider. As in the case of full uncertainty, this valuation can always be assigned probability $1$ by the Dirac measure, the divider's worst-case belief will be to always face the deteriorating action of the chooser. For instance, if the divider announces a low (high) price $p$ with $x_D-p > p$ ($x_D-p<p$) then a worst-case distribution is to assign probability $1$ to any valuation $x_C^* > 2p$ ($x_C^*<2p$), e.g. $x_C^* = x_D$, resulting in the chooser buying the divider's shares (selling her shares), which coincides with the chooser's maxmin action. Hence, although different in their philosophy\footnote{While the maxmin setting considers a mischievous opponent (no matter their valuation), the full uncertainty setting deals with worst-case beliefs about the opponent's valuation.}, both settings amount to the behavior of the divider fully-hedging herself.

This observation lets us locate the maxmin and the Bayesian price announcement on two opposite ends of a setting of Knightian Uncertainty: One of full uncertainty about the peer's CDF (maxmin price) and one of no uncertainty, where only a single CDF is being faced (Bayesian price). The following section reveals the optimal price announcements for intermediate levels of uncertainty to be a combination of the former two.

\subsection{Optimal Price Announcement}\label{Section: Optimal Price Announcement}

In contrast to the chooser, the divider faces a delicate trade-off when announcing her price if Knightian Uncertainty in the form of a distribution band $\cG$ is dealt with. Nevertheless, the divider's maximization problem can be explicitly solved and is our main result, establishing and revealing a link between the Bayesian setting and the one of full uncertainty which are recovered as special cases. In fact, the optimal price announcement will be a combination of the optimal price announcements encountered in Section \ref{Section: Bayes and maxmin}: An uncertainty averse divider will announce full-hedging prices if and only if she considers her valuation to be between the smallest and largest median admissible in $\cG$. For a valuation $x_D$ outside this interval, she will announce the respective Bayesian price corresponding to the stochastically dominant (high $x_D$) or stochastically dominated (low $x_D$) representative of $\cG$.

We allow for some more generality in the regularity assumptions on the lower and upper bounds $G_0,G_1$ of the distribution band than in Section \ref{Section: Bayes and maxmin}. Let $G_0,G_1$ be distribution functions on $X$ that are continuous on $[\xl,\xu)$\footnote{As usual, distribution functions are assumed to be right-continuous, thus, continuity from the right in $\xl$ is immediate. However, there might still be jumps in $\xl$ (i.e., $G_i(\xl)>0$, putting a point-mass on $\xl$) and, of course, for $G_0$, in $\xu$.} and piecewise continuously differentiable with positive density and all left and right limits on $[\xl,\xu]$. More precisely, there are finitely many points in $X$ such that on each partition, $G_i$ is continuously differentiable and left and right limits of $G_i$ and their derivatives always exist. We further assume that $G_1$ is continuous in $\xu$ and thus everywhere.
These conditions are automatically fulfilled, if we start with an atom-less distribution function $F$ with positive density as in Section \ref{Section: Bayes and maxmin} and consider $\cG(F,\epsilon)$.\footnote{The same is true when considering a Prohorov-ball around $F$}

Note that $\pi$ is continuous, not only in $x_D$ but also in $p$ for $2p \in [\xl,\xu]$: The possible discontinuity of $\pi_{G_0}(.\mid x_D)$ in $2p=\xu$ only occurs if $x_D <\xu$. But in that case, $\pi(. \mid x_D)$ is equal to the continuous function $\pi_{G_1}( . \mid x_D)$ around $2p = \xu$.

Thus, the optimal price announcement (correspondence) can be defined
\begin{equation}
m(x_D) := \arg \max_p \pi (p \mid x_D).
\end{equation}

Under Assumption \ref{AssGeneral}, the main driver of our main theorem, it will turn out, that the maximizer is unique, so the optimal price announcement under uncertainty is indeed a function rather than a correspondence.

\begin{assumption}\label{AssGeneral}
For a considered valuation $x_D \in X$, let both the divider's payoff functions $\pi_{G_0}(p \mid x_D)$ and $\pi_{G_1}(p \mid x_D)$ be strictly quasiconcave in $p$ for $2p \in [\xl,\xu]$.
\end{assumption}

The notion of a distribution band $\cG(F,\epsilon)$ is compatible with Assumption \ref{AssGeneral}:

\begin{lemma}\label{Lemma: suff epsilon}
If $F$ fulfills the SHRC (Lemma \ref{Lemma: Suff McAfee}), $G_0$ and $G_1$ from $\cG(F,\epsilon)$ induce payoff functions fulfilling Assumption \ref{AssGeneral}.
\end{lemma}

\begin{example} \label{Example: SHRC}
CDFs fulfilling Lemma \ref{Lemma: suff epsilon} include triangular distributions, the truncated standard normal distribution on $[0, 1]$ or the class of beta distributions.

However, the sufficient condition is far from being sharp: In the case of risk neutrality, i.e. $u = \id$, e.g., some beta distributions induce strictly quasiconcave utility functions although the sufficient condition is not satisfied, see Figure \ref{Fig: Beta example}.
\begin{figure}
    \centering
    \includegraphics[scale=0.4]{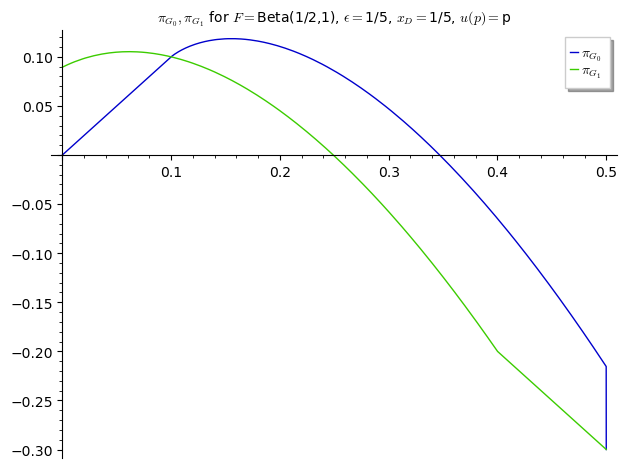} \qquad \includegraphics[scale=0.4]{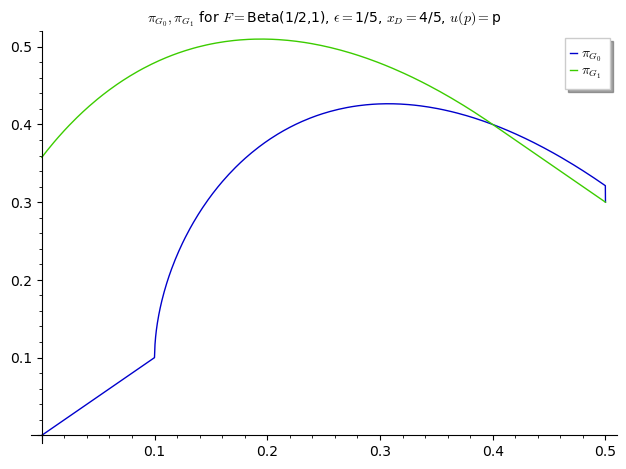}
    \caption{Strict quasiconcavity of the payoff functions under the beta distribution $B(\tfrac{1}{2},1)$ for different values of $x_D$ and $\epsilon$.}
    \label{Fig: Beta example}
\end{figure}
\end{example}


We illustrate the Texas Shoot-Out under Knightian Uncertainty by considering a version of Example \ref{Example: no uncertainty} that uses a distribution band. It turns out that, on the one hand, a divider with a rather average valuation will play cautiously and fully-hedge herself against any losses by playing half her valuation. On the other hand, a divider with very low (resp. high) valuation will still try to strategically extract revenue by stating a slightly higher (resp. lower price), thinking that the chooser will still take the offer (resp. refuse) it.

\begin{example}\label{Example: uniform distribution}
Let $F$ be the uniform distribution on the unit interval $[0,1]$ and let $u= \id$. For any $\epsilon \leq 1/2$ the distributions defining the distribution band $\cG(F,\epsilon)$ are given by
\begin{align*}
G_0(x) 
= \begin{cases}
0 &, 0\leq x \leq \epsilon,\\
x-\epsilon &, \epsilon < x < 1,\\
1 &, x=1.
\end{cases}
\quad \text{and} \quad 
G_1(x) 
= \begin{cases}
x+\epsilon &, 0 \leq x \leq 1- \epsilon\\
1 &, 1-\epsilon < x \leq 1.
\end{cases}.
\end{align*}
The  optimal price announcement can be explicitly calculated and is given by
\begin{align*}
m(x_D) = \begin{cases}
\frac{x_D}{4} - \frac{\epsilon}{4} + \frac{1}{8} &, 0\leq x_D < \frac{1}{2} - \epsilon,\\
\frac{x_D}{2} &, \frac{1}{2}-\epsilon \leq x_D \leq \frac{1}{2} + \epsilon,\\
\frac{x_D}{4} + \frac{\epsilon}{4} + \frac{1}{8} &, \frac{1}{2} + \epsilon < x_D \leq 1,
\end{cases}
\end{align*}
which is depicted in Figure \ref{Figure: uniform price announcement}.

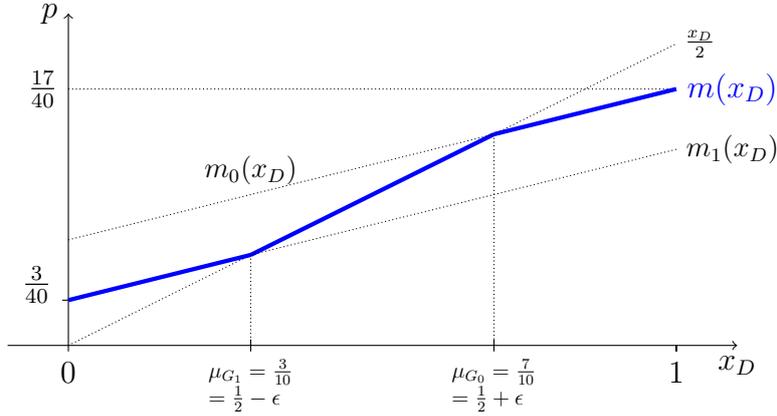
\begin{figure}[ht]
\begin{center}
\begin{tikzpicture}[scale=8]
\draw[<->] (0,0.55)node[left]{$p$}--(0,0)--(1,0)--(1,-0.01)node[below]{$1$}--(1,0)--(1+1/10,0)node[below]{$x_D$};
\draw (-1/10,0)--(0,0);
\draw (0,0)--(0,-1/100)node[below]{$0$};


\draw (0,3/40)--(-1/100,3/40)node[left,yshift=0.5em]{$\tfrac{3}{40}$};
\draw[densely dotted] (0,17/40)node[left]{$\tfrac{17}{40}$}--(1,17/40)node[right,blue]{$m(x_D)$};
\draw[densely dotted](0,0)--(1,1/2)node[right,scale=0.8]{$\tfrac{x_D}{2}$};

\draw[densely dotted](3/10,3/20)--(3/10,0);
\draw[densely dotted](7/10,7/20)--(7/10,0);
\draw (3/10,0.01)--(3/10,-0.01)node[below,scale=0.7,align=left]{$\mu_{G_1} =\tfrac{3}{10}$\\ $=\frac{1}{2} -  \epsilon$};
\draw (7/10,0.01)--(7/10,-0.01)node[below,scale=0.7,align=left]{$\mu_{G_0} = \tfrac{7}{10}$\\$=\frac{1}{2} +  \epsilon$};

\draw[scale=1, domain=3/10:7/10, smooth, ultra thick, variable=\x, blue]  plot ({\x}, {1/2 * \x});
\draw[scale=1, domain=0:3/10, smooth, ultra thick, variable=\x, blue]  plot ({\x}, {1/4 * \x + 1/8 - 1/20});
\draw[densely dotted] (3/10,1/4*3/10+1/8-1/20)--(1,1/4*1+1/8-1/20)node[right,scale=0.9]{$m_1(x_D)$};
\draw[scale=1, domain=7/10:1, smooth,ultra thick, variable=\x, blue]  plot ({\x}, {1/4 * \x + 1/8 + 1/20});
\draw[densely dotted] (0,0+1/8+1/20)--(3/10,1/4*3/10+1/8+1/20)node[above,scale=0.9]{$m_0(x_D)$}--(7/10,1/4*7/10+1/8+1/20);
\end{tikzpicture}
\caption{Optimal price announcement $m(x_D)$ for $\cG(F,\epsilon)$ where $F \sim \mathcal{U}([0,1]), u = \id$ and $\epsilon = \tfrac{1}{5}$. The functions $m_i(x_D)$ ($i=0,1$) indicate the optimal price announcement in the Bayesian setting if, resp., the CDF $G_i$ is faced.}
\label{Figure: uniform price announcement}
\end{center}
\end{figure}
For $\epsilon \geq 1/2$ the divider will play the full-hedging strategy $m(x_D) = \tfrac{x_D}{2}$ for all her valuations.

It is worth mentioning that full uncertainty, i.e. $\cG(F,\epsilon)$ containing all CDFs on $X$, is only faced for $\epsilon \geq 1$. Complementary to \cite{van2020}, this indicates that full uncertainty is sufficient, but not necessary for full-hedging to be optimal for every valuation.
\end{example}

In order to give a precise and complete game theoretic description of the Texas Shoot-Out under multiple priors, we also allow for the chooser to face uncertainty about the distribution of the divider's valuation. Since she has all her payoff-relevant information at hand when choosing an action, her belief about the divider's valuation will only play a role when talking about interim worst-case expected utility in Section \ref{Section: interim utility}.

For the following equilibrium concept, assume a symmetric setting in which the chooser faces the same Knightian Uncertainty about the divider's valuation given by the set $\cG$ of distribution functions. The game is then again modeled similar to a Bayesian one with Gilboa-Schmeidler preferences. We take the set $\cG \otimes \cG$ to contain independent draws\footnote{See \cite{muraviev} for a precise description. Also note that we can fit our game into their framework by adding nature as a player to play an ambiguous \emph{Ellsberg strategy} as in \cite{riedel2014ellsberg}.} of distribution functions $G_D,G_C \in \cG$, so that, after valuations are privately observed, both agents face the prior set $\cG$ about the peer's valuation distribution function. Ex-ante pure\footnote{Picking the favorable pure strategy of chooser that is always a best reply, our assumptions imply pure price announcements. Thus, there is no need for considering mixed strategies.} strategies are denoted by $s_D \colon X \to \R, s_C \colon X \times \R \to \{ \text{buy, sell} \}$. The expressions $\widetilde{\pi}_D, \widetilde{\pi}_C$ stand for the realized payoff of the divider/chooser at the end of the mechanism, given the agents' valuation profiles and their strategies. The associated ex-ante worst-case expected payoff functions can be written as
\begin{align*}
    U_D^{\text{ex}}(s_D,s_C) &= \min_{(G_D,G_C) \in \cG \otimes \cG } \E_{G_D \otimes G_C} \left[ \widetilde{\pi}_D(s_D,s_C; x_D,x_C) \right]\\
    U_C^{\text{ex}}(s_D,s_C) &= \min_{(G_D,G_C) \in \cG \otimes \cG } \E_{G_D \otimes G_C} \left[ \widetilde{\pi}_C(s_D,s_C; x_D,x_C) \right].
\end{align*}
Interim worst-case expected utility is given by
\begin{align*}
    U_D^{\text{int}}(p,s_C \mid x_D) &= \min_{G_C \in \cG } \E_{G_C} \left[ \widetilde{\pi}_D(s_D,s_C; x_D,x_C) \right]\\
    U_C^{\text{int}}(a,s_D \mid x_C) &= \min_{G_D \in \cG } \E_{G_D \mid s_D(x_D) = p} \left[ \widetilde{\pi}_C(s_D,s_C; x_D,x_C) \right]
\end{align*}
From that it is clear that the decision rule 'sell' if and only if $x_C - p \leq p$ is always a best reply ex-ante and interim for any announced and observed price $p = s_D(x_D)$ independent of how the divider has come to set this price, i.e. no matter what updated information the chooser has on $x_D$ or $G_D$. Especially, no \emph{dynamic inconsistency}\footnote{\emph{Dynamic inconsistent} behavior is characterized by agents having an incentive to deviate from their ex-ante plan and usually occurs from updating by Bayes rule in environments with Knightian Uncertainty, see \cite{epstein2003recursive}, \cite{aryal2014note}, \cite{pahlke2022dynamic}, \cite{hanany2020incomplete} for how this impacts extensive form games. An established way to avoid this is by introducing \emph{rectangular beliefs}.} arises. If the divider anticipates this strategy of the chooser, her interim optimal behavior is also ex-ante optimal, thus it suffices to derive the interim best response for the following equilibrium concept:

\begin{definition}\label{Definition: ex-ante and interim Knight equilibrium}
A strategy profile $(s_D,s_C)$ is an \emph{ex-ante Knight-Nash equilibrium} if it is a Nash equilibrium for the game with ex-ante worst-case expected payoff functions.\\
A strategy profile $(s_D,s_C)$ is an \emph{interim Knight-Nash equilibrium} if it is a Nash equilibrium for every game induced by $x_D$ and $(x_C,p)$ with interim worst-case expected payoff functions.
\end{definition}

By continuity of $\pi(. \mid x_D)$ we can always find an optimizer, thus we have

\begin{proposition}\label{Proposition: Knight-Nash existence}
In the Texas Shoot-Out, the following ex-ante and interim Knight-Nash equilibrium always exists. For any price $p$, a chooser with valuation $x_C$ will sell the company if and only if $x_C \leq 2p$. A divider with valuation $x_D$ will make an optimal price announcement $p$ belonging to $m(x_D)$.
\end{proposition}

Strict quasiconcavity is driving our main result - a full characterization of the optimal price announcement of the divider. A link between the price announcements from \cite{mcafee92} and \cite{van2020} for intermediate levels of uncertainty is established: We explicitly describe the transformation of the optimal price announcement from the stochastic case (single prior or $\epsilon = 0$) to the setting of full uncertainty (full set of priors or $\epsilon \gg 0$) by two observations: Firstly, the divider will hedge herself for valuations between the medians. Secondly, for low resp. high valuations, the divider will make the optimal price announcement corresponding to facing the distribution $G_1$ resp. $G_0$. Note that for increasing levels of uncertainty (increasing $\epsilon$) the spread of the median increases.

\begin{theorem}\label{Theorem:divider characterization}
The optimal price announcement of a divider with valuation $x_D$ is given by
\begin{align*}
    m(x_D) =
    \begin{cases}
    m_1(x_D) &, x_D < \mu_{G_1}^-,\\
    \frac{x_D}{2} &, \mu_{G_1}^- \leq x_D \leq \mu_{G_0}^+,\\
    m_0(x_D) &, \mu_{G_0}^+ < x_D,
    \end{cases}
\end{align*}
where $m_i(x_D) := m_{\pi_{G_i}(.\mid x_D)}$ denotes the unique point of the strictly quasiconcave function $\pi_{G_i}(. \mid x_D)$ separating its strictly decreasing and increasing domains and $\mu_{G_1}^-$ the smallest and $\mu_{G_0}^+$ the largest median of $G_1$ resp. $G_0$.
\end{theorem}

Theorem \ref{Theorem:divider characterization} admits the following additional observations encoding a well-behaved price announcement.

Firstly, the optimal price announcement will be above half of the valuation while staying below the full-hedging action below the median of $\mu_{G_1}$ and vice versa above $\mu_{G_0}$, revealing that the divider expects to extract payoff from her belief that the chooser has a higher resp. lower valuation than herself. By \ref{Theorem:divider characterization} she will not do so for intermediate valuations, but rather stick to a safe offer.

\begin{cor}\label{Corollary: B xm and fix points}
We have $x_D < 2m(x_D) < \mu_{G_1}^-$ if $x_D < \mu_{G_1}^-$ and $\mu_{G_0}^+ < 2m(x_D) < x_D$ if $\mu_{G_0}^+ < x_D$.
\end{cor}

Secondly, she will announce higher prices for higher valuations.

\begin{cor} \label{Corollary: B cont and str increasing}
The optimal price announcement $m(x_D)$ is increasing and thus a.s. continuous and measurable. It is continuous and strictly increasing for $\mu_{G_1}^- \leq x_D \leq \mu_{G_0}^+$ and on regions where $G_0$ resp. $G_1$ are continuously differentiable with positive derivative in $m_0(x_D)$ resp. $m_1(x_D)$.\\
Furthermore, $m(x_D)$ is always continuous in $x_D \in \{ \mu_{G_1}^-, \mu_{G_0}^+ \}$.
\end{cor}


\subsection{Efficiency}\label{Section: Efficiency}
Only one co-owner remains in the company after the Texas Shoot-Out mechanism has been completed. Desirably, the agent with the highest valuation becomes the sole owner. In line with \cite{mcafee92}, we then call outcome \emph{allocatively efficient}.

Our previous results determine exactly under which circumstances the Texas Shoot-Out under Knightian Uncertainty yields an allocatively efficient outcome: Recall that the chooser sells her shares if and only if $2p \geq x_D$. Now, on the one hand side, for $x_D \in [\mu_{G_1}^-, \mu_{G_0}^+]$ the divider announces $2p^* =2m(x_D)= x_D$ by Theorem \ref{Theorem:divider characterization} and consequently, the divider obtains the good if and only if $x_D \geq x_C$. On the other hand side, for $x_D < \mu_{G_1}^-$ the divider announces a price fulfilling $2p^* = 2m(x_D)>x_D$ by Corollary \ref{Corollary: B xm and fix points} leading to an inefficient allocation if and only if $x_C \in [x_D,2m(x_D)]$, where the divider obtains the company while having a lower valuation than the chooser. Similarly, for $x_D > \mu_{G_0}^+$ inefficiencies arise if and only if $x_C \in [2m(x_D),x_D]$. Figure \ref{Figure: inefficiencies uniform case} depicts the profiles $(x_D,x_C)$ for which the Texas Shoot-Out generates inefficient locations in the settings of Example \ref{Example: no uncertainty} and Example \ref{Example: uniform distribution}.

\begin{figure}[ht]
    \centering

    \begin{tikzpicture}[scale=4.5]
    \draw[<->](0,1.1)node[left]{$x_C$}--(0,1)--(-0.01,1)node[left]{1}--(0,1)--(0,1/2)--(-0.01,1/2)node[left]{$\tfrac{1}{2}$}--(0,1/2)--(0,0)node[below]{$0$}--(1/2,0)--(1/2,-0.01)node[below]{$\tfrac{1}{2}$}--(1/2,0)--(1,0)--(1,-0.01)node[below]{$1$}--(1,0)--(1.1,0)node[below]{$x_D$};
    \draw (0,1)--(1,1)--(1,0);
    \draw[densely dotted](0,0)--(1,1);
    \draw (0,1/4)--(-0.01,1/4)node[left]{$\tfrac{1}{4}$};
    \draw (1,3/4)--(1.01,3/4)node[right]{$\tfrac{3}{4}$};
    \filldraw[red,opacity=0.2](0,1/4)--(0,0)--(1,1)--(1,3/4)--cycle;
    \draw (1/2,-3/10)node[]{no uncertainty};

    \begin{scope}[xshift=1.5cm]
    \draw[<->](0,1.1)node[left]{$x_C$}--(0,1)--(-0.01,1)node[left]{1}--(0,1)--(0,0)node[below]{$0$}--(1,0)--(1,-0.01)node[below]{$1$}--(1,0)--(1.1,0)node[below]{$x_D$};
    \draw (0,1)--(1,1)--(1,0);
    \draw[densely dotted](0,0)--(1,1);
    \draw (0,3/20)--(-0.01,3/20)node[left]{$\tfrac{3}{20}$};
    \draw (1,17/20)--(1.01,17/20)node[right]{$\tfrac{17}{20}$};
    
    \draw[densely dotted] (3/10,3/10)--(3/10,-0.01)node[below]{$\tfrac{3}{10}$};
    \draw[densely dotted] (7/10,7/10)--(7/10,-0.01)node[below]{$\tfrac{7}{10}$};
    
    \filldraw[red,opacity=0.2](0,3/20)--(0,0)--(3/10,3/10)--cycle;
    \filldraw[red,opacity=0.2] (1,1)--(1,17/20)--(7/10,7/10)--cycle;

    \draw (1/2,-3/10)node[]{uncertainty};
    \end{scope}
    
    \end{tikzpicture}
    \caption{Valuation profiles for which the Texas Shoot-Out yields inefficient allocations in red for the uniform distribution $F$ in the settings of Example \ref{Example: no uncertainty} (no uncertainty) and Example \ref{Example: uniform distribution} (uncertainty for $\cG(F,\tfrac{1}{5})$).}
    \label{Figure: inefficiencies uniform case}
\end{figure}
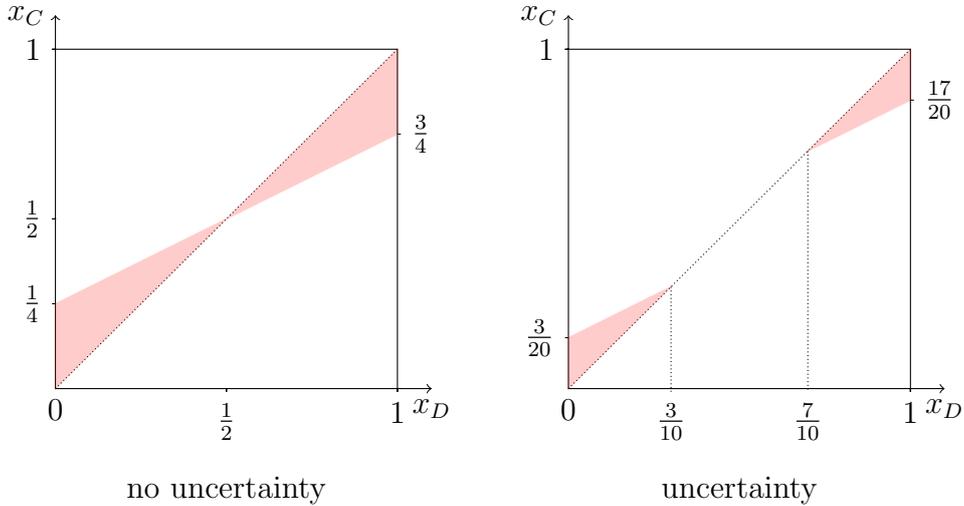

Introducing Knightian Uncertainty to the Texas Shoot-Out mechanism increases allocatively efficiency in comparison to the Bayesian case. More valuation profiles will lead to an efficient allocation, especially all those for which the divider's valuation is close to the believed (uncertain) median of the chooser's valuation. But also for extremely low or high valuations, an increase in Knightian Uncertainty will increase efficiency by weakening the divider's bargaining power, leading to her announcing prices closer to half her valuation.

\subsection{Interim Utility}\label{Section: interim utility}

Having a valuation $x$ and modeling uncertainty about the peer's valuation by a distribution band $\cG$, we can ask whether an agent prefers to be the divider (and initiate the Texas Shoot-Out) or the chooser. Focusing on a risk neutral agent, we will see that there always is a preference for being the chooser (rather waiting for the peer to initiate it). This preference is strict for all valuations if the uncertainty is small, while for high levels of uncertainty it is strict only for low/high valuations, unless full uncertainty is faced.

We assume some more regularity on the distribution functions $G_0,G_1$. Put simply, we require them to be twice continuously differentiable with positive density where they are not $0$ or $1$. In more detail, in addition to $G_1$ being continuous and $G_0$ being continuous on $[\xl,\xu)$, we assume that there exists values $\a,\b \in X$ such that $G_0$ is zero on $[\xl,\a]$, twice continuously differentiable with positive density on $[\a,\xu)$ and $G_1$ is twice continuously differentiable with positive density on $[\xl, \b]$ and equal to $1$ on $[\b,\xu]$. Especially, they will each have a unique median $\mu_{G_0}, \mu_{G_1}$. Note that this is automatically fulfilled if we consider $\cG(F,\epsilon)$ around an atom-less CDF $F$ that is twice continuously differentiable with positive density on $X$.\bigskip

By Theorem \ref{Theorem:divider characterization} interim worst-case EU of a divider with valuation $x_D$ can be immediately written down to be
\begin{equation*}
    \Phi_D (x_D) := \pi(m(x_D) \mid x_D) = 
    \begin{cases}
    \pi_{G_1}(m_{1}(x_D) \mid x_D) &,  x_D < \mu_{G_1},\\
    \frac{x_D}{2} &, \mu_{G_1} \leq x_D \leq \mu_{G_0},\\
   \pi_{G_0}(m_0(x_D) \mid x_D) &, \mu_{G_0} < x_D.
    \end{cases}
\end{equation*}

So far, we have not investigated the chooser's interim expected utility, but will do so now. To this end, assume a symmetric setting where the chooser faces the same uncertainty about the peer's valuation as the divider, i.e. she beliefs the divider's valuation to be drawn from the set $\cG$ while having own valuation $x_C$. Her interim worst-case expected payoff in equilibrium is given by
\begin{align*}
    \Phi_C(x_C) := \min_{G \in \cG} \mathbb{E}_{G} \left[ \max\{ u(x_C - m(z)), u(m(z)) \} \right],
\end{align*}
where the expectation is taken w.r.t. the CDF $G$ over the divider's valuations, denoted by the variable $z$, and $m(z)$ is the divider's optimal price announcement if her valuation is $z$.

Since $2m(z)$ is increasing with range $[2m(\xl),2m(\xu)]$ by Corollary \ref{Corollary: B cont and str increasing} we can draw the following conclusion about the worst-case beliefs of the chooser.

\begin{lemma}\label{Lemma: Worst-case distributions chooser}
Worst-case beliefs about the divider's valuation for a chooser with valuation $x_C$ are given by
\begin{equation*}
    \begin{cases}
    G_1 &, x_C<2m(\xl),\\
    G^*_{x_C} &, 2m(\xl) \leq x_C \leq 2m(\xu),\\
    G_0 &, 2m(\xu)<x_C,
    \end{cases}
\end{equation*}
where we define
\begin{equation*}
    G^*_{x_C}(z) = \begin{cases}
        G_0(z) &, \xl \leq z < z_*,\\
        G_1(z) &, z_* \leq z \leq \xu
        \end{cases}
\end{equation*}
for the unique $z_*$ fulfilling $2m(z_*) = x_C$.
\end{lemma}

As the chooser can always pick her favorite piece of the cake among $u(x_C-p)$ and $u(p)$, she is worse off, the closer $p$ is to $\tfrac{x_C}{2}$. Thus, if there is a valuation $z_*$ of a divider that induces the price $2m(z_*) = x_C$, the chooser's worst-case belief is given by $G^*_{x_C}$, i.e. the distribution in $\cG$ that puts most weight at and around $z_*$.We depicted an example of $G^*_{x_C}$ in Figure \ref{Figure: G star}.

\begin{figure}[ht]
\begin{center}
\begin{tikzpicture}[scale=6]
\draw[->] (-0.1,0)--(1,0)node[below]{$1$}--(1.2,0)node[below]{$x$};
\draw[->] (0,0)node[below]{$0$}--(0,1)node[left]{$1$}--(0,1.1)node[left]{$\mathbb{P}$};
\draw[densely dotted] (0,1)--(1,1);
\draw[densely dotted] (1,0)--(1,1);

\draw[dashed, thick, blue] (0,0)--(1/5,0)--(1,4/5);
\draw[dashed, thick, blue] (0,1/5)--(4/5,1)--(1,1);
\draw [blue,thick,fill=blue] (0,1/5) circle(0.25pt);
\draw[blue,thick,fill=white] (1,4/5) circle(0.25pt);
\draw[blue] (0.85,0.61)node[right]{$G_0$};
\draw[blue] (0.26,0.5)node[left]{$G_1$};
\draw[red] (0.4,0.25)node[left,scale=1.3]{$G^*_{x_C}$};

\draw[ultra thick, red] (0,0)--(1/5,0)node[below,black]{$x^0 = \tfrac{1}{5}$}--(1/2,3/10);
\draw[densely dotted](4/5,1)--(4/5,0)node[below]{$x^1 = \tfrac{4}{5}$};
\draw[red, ultra thick, densely dotted] (1/2,3/10)--(1/2,7/10);
\draw[densely dotted](1/2,3/10)--(1/2,0)node[below]{$z_*$};
\filldraw[red, fill=white] (1/2,3/10) circle(0.35pt);
\filldraw[red, fill=red] (1/2,7/10) circle(0.35pt);
\draw[ultra thick,red](1/2,7/10)--(8/10,1)--(1,1);
\draw[red,fill=red] (1,1) circle(0.35pt);

\end{tikzpicture}
\end{center}
\caption{Illustration of the worst-case distribution $G^*_{x_C}$ of the chooser in case of $2m(z_*) = x_C$ given $\cG(\mathcal{U}([0,1]),\tfrac{1}{5})$. $G^*_{x_C}$ puts most weight on the valuation $z_*$ of the divider, that leads to the worst offer for the chooser with valuation $x_C$. Here $x_C = \tfrac{1}{2}$ and thus $z_* = \tfrac{1}{2}$ by Example \ref{Example: uniform distribution}}
\label{Figure: G star}
\end{figure}
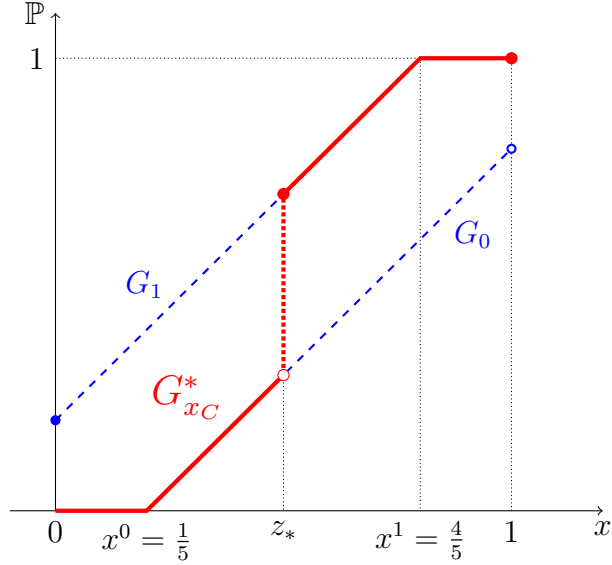

In the following, we will restrict to the case of risk neutral agents only, i.e. $u = \id$. A useful tool in preparation of the interim EU comparison and  for explicit calculations is the characterization of the derivatives of $\Phi_D$ and $\Phi_C$.

\begin{lemma}\label{Lemma: derivatives of interim expected utility}
Let $u = \id$. The functions $\Phi_D$ and $\Phi_C$ are continuous, increasing and piecewise differentiable with derivatives
\begin{align*}
    \Phi'_D(x) &= 
    \begin{cases}
    G_1(2m_1(x)) &,  x < \mu_{G_1},\\
    \frac{1}{2} &, \mu_{G_1} \leq x \leq \mu_{G_0},\\
    G_0(2m_0(x)) &, \mu_{G_0} < x.
    \end{cases}\\
    &\text{and}\\
    \Phi'_C(x) &= 
    \begin{cases}
    0&, \xl \leq x < 2m(\xl),\\
    \frac{1}{2} \cdot \left( G_1(m^{-1}(\tfrac{x}{2})) + G_0(m^{-1}(\tfrac{x}{2})) \right) &, 2m(\xl) \leq x \leq 2m(\xu)\\
    1 &, 2m(\xu) < x \leq \xu.
    \end{cases}
\end{align*}
\end{lemma}

Thus, the higher the valuation one has, the higher the interim worst-case expected utility. In fact, from the functional form of the derivatives, the utility strictly increases with increasing valuation except for a chooser facing very low valuations.\\ 

By means of the above lemma we can show that one is never worse off being the chooser and can characterize precisely when this preference is strict: One strictly prefers to be the chooser for all valuations for small levels of uncertainty. This observation stays true for very low/high valuations if uncertainty is increased until full uncertainty is faced.

\begin{theorem}\label{Theorem: interim EU comparison}
If $\a<\b$, we have $\Phi_D(x)<\Phi_C(x)$ for all $x \in X$. If $\b \leq \a$, $\Phi_D(x) = \Phi_C(x) = \tfrac{x}{2}$ for $x \in [\b, \a]$ and $\Phi_D(x) < \Phi_C(x)$ otherwise.
\end{theorem}

Especially, under full uncertainty both, the divider and the chooser, have an interim worst-case expected utility equal to half their valuation.

We conclude this section by examining our uniform example for an interim utility comparison.
\begin{example}\label{Example: interim wc EU}
Consider again $F \sim \mathcal{U}([0,1])$ for risk-neutral agents, i.e. $u = \id$. Since we have already calculated $m(x)$, we can calculate $\pi(m(x) \mid x)$ for a divider with valuation $x$, obtaining her interim worst-case expected utility
\begin{align*}
    \Phi_D(x) = \begin{cases}
    \frac{1}{4}x^2 + (\frac{1}{4} + \tfrac{\epsilon}{2}) \cdot x + \frac{\epsilon^2}{4} - \frac{\epsilon}{4} + \frac{1}{16} &,0 \leq x < \frac{1}{2} - \epsilon\\
    \frac{x}{2} &, \frac{1}{2} - \epsilon \leq x \leq \frac{1}{2} + \epsilon,\\
    \frac{1}{4}x^2 + (\frac{1}{4} - \frac{\epsilon}{2}) \cdot x + \frac{\epsilon^2}{4} + \frac{\epsilon}{4} + \frac{1}{16} &, \frac{1}{2} + \epsilon < x \leq 1,
    \end{cases}
\end{align*}
where for $\epsilon > \tfrac{1}{2}$ the function should be read as being equal to the function $\tfrac{x}{2}$ everywhere.

In order to calculate $\Phi_C$ we need to distinguish three cases $0 \leq \epsilon \leq \tfrac{1}{4}$, $\tfrac{1}{4}< \epsilon \leq \tfrac{1}{2}$ and $\tfrac{1}{2} < \epsilon \leq 1$ since not only the pasting points for $m$ (or $m^{-1}$) play a role, but also the kinks of $G_0, G_1$.\footnote{If $\epsilon \leq \tfrac{1}{4}$, we have $\epsilon \leq \tfrac{1}{2}-\epsilon \leq \tfrac{1}{2}+\epsilon \leq 1-\epsilon$ while for $\tfrac{1}{4}<\epsilon \leq \tfrac{1}{2}$ we have $\tfrac{1}{2}-\epsilon < \epsilon \leq 1-\epsilon < \tfrac{1}{2}+\epsilon$ and $\tfrac{1}{2} < \epsilon$ implies  $\tfrac{1}{2}-\epsilon <  1-\epsilon < \epsilon < \tfrac{1}{2}+\epsilon$.} Also note that for $\epsilon \geq 1$ we have full uncertainty and the situation won't change anymore. All the explicit formulae are delegated to the appendix.

Figure \ref{Figure: interim wc EU comparison} graphically illustrates $\Phi_D$ and $\Phi_C$ for $\epsilon = 0.02,0.4,0.6$.

\begin{figure}[ht]
    \centering
    \begin{tikzpicture}[scale=4.2]
    \draw[<->](0,0.8)node[left]{EU}--(0,0.5)node[left]{$0.5$}--(-0.01,0.5)--(0,0.5)--(0,0)node[left]{$0$}node[below]{$0$}--(1,0)node[below]{$1$}--(1,-0.01)--(1,0)--(1.2,0)node[below]{$x$};
    \draw (0,0.24485)node[left,scale=0.75]{$\approx 0.25$};
    
    \draw (1.06,0.75)node[red]{$\Phi_C$};
    \draw (1.06,0.56)node[blue]{$\Phi_D$};
    
    
    \draw[scale=1, domain=0:12/25, smooth, ultra thick, variable=\x, blue] plot ({\x}, {1/4*\x*\x + 13/50*\x + 36/625});
    \draw[scale=1, domain=12/25:13/25, smooth, ultra thick, variable=\x, blue] plot ({\x}, {1/2*\x});
    \draw[scale=1, domain=13/25:1, smooth, ultra thick, variable=\x, blue] plot ({\x}, {1/4 * \x*\x +6/25*\x +169/2500});
    
    \draw[scale=1, domain=0:6/25, smooth, ultra thick, variable=\x, red] plot ({\x}, {0.24485});
    \draw[scale=1, domain=6/25:1/4, smooth, ultra thick, variable=\x, red] plot ({\x}, {1/2*\x*\x - 23/100*\x + 217/800});
    \draw[scale=1, domain=1/4:12/25, smooth, ultra thick, variable=\x, red] plot ({\x}, {\x*\x-12/25*\x + 121/400});
    \draw[scale=1, domain=12/25:13/25, smooth, ultra thick, variable=\x, red] plot ({\x}, {1/2*\x*\x + 1873/10000});
    \draw[scale=1, domain=13/25:3/4, smooth, ultra thick, variable=\x, red] plot ({\x}, {\x*\x -13/25*\x + 129/400});
    \draw[scale=1, domain=3/4:19/25, smooth, ultra thick, variable=\x, red] plot ({\x}, {1/2*\x*\x + 23/100*\x + 33/800});
    \draw[scale=1, domain=19/25:1, smooth, ultra thick, variable=\x, red] plot ({\x}, {\x-0.25515});
    
    \draw[thin, densely dotted] (0,0)--(1,0.5)node[below]{$\frac{x}{2}$};
    \draw (1/2,-0.2)node{$\epsilon = 0.02$};

    \begin{scope}[xshift=1.75cm]
    \draw[<->](0,0.8)node[left]{EU}--(0,0.5)node[left]{$0.5$}--(-0.01,0.5)--(0,0.5)--(0,0)node[left]{$0$}node[below]{$0$}--(1,0)node[below]{$1$}--(1,-0.01)--(1,0)--(1.2,0)node[below]{$x$};
    \draw(0,81/800)node[left,scale=0.75]{$\approx 0.1$};
    
    \draw (1.04,0.6)node[red,right]{$\Phi_C$};
    \draw (1.04,0.5)node[blue,right]{$\Phi_D$};
    
    \draw[scale=1, domain=0:1/10, smooth, ultra thick, variable=\x, blue] plot ({\x}, {1/4 * \x*\x + 9/20*\x + 1/400});
    \draw[scale=1, domain=1/10:9/10, smooth, ultra thick, variable=\x, blue] plot ({\x}, {1/2 * \x});
    \draw[scale=1, domain=9/10:1, smooth, ultra thick, variable=\x, blue] plot ({\x}, {1/4 * \x*\x + 1/20 * \x + 81/400});
    
    
    \draw[scale=1, domain=0:1/20, smooth, ultra thick, variable=\x, red] plot ({\x}, {81/800});
    \draw[scale=1, domain=1/20:1/10, smooth, ultra thick, variable=\x, red] plot ({\x}, {1/2 * \x*\x +3/20*\x + 37/400});
    \draw[scale=1, domain=1/10:2/5, smooth, ultra thick, variable=\x, red] plot ({\x}, {1/4 *\x*\x + 1/5*\x + 9/100});
    \draw[scale=1, domain=2/5:3/5, smooth, ultra thick, variable=\x, red] plot ({\x}, {1/2*\x*\x +13/100});
    \draw[scale=1, domain=3/5:9/10, smooth, ultra thick, variable=\x, red] plot ({\x}, {1/4*\x*\x + 3/10*\x + 1/25});
    \draw[scale=1, domain=9/10:19/20, smooth, ultra thick, variable=\x, red] plot ({\x}, {1/2 *\x*\x -3/20*\x + 97/400});
    \draw[scale=1, domain=19/20:1, smooth, ultra thick, variable=\x, red] plot ({\x}, {\x-319/800});
    \draw (1/2,-0.2)node{$\epsilon = 0.4$};
    \end{scope}
    
    \begin{scope}[yshift=-1cm]
    \draw[<->](0,0.8)node[left]{EU}--(0,0.5)node[left]{$0.5$}--(-0.01,0.5)--(0,0.5)--(0,0)node[left]{$0$}node[below]{$0$}--(1,0)node[below]{$1$}--(1,-0.01)--(1,0)--(1.2,0)node[below]{$x$};
    
    \draw[densely dotted] (0.4,0)node[below,scale=0.8]{$0.4$}--(0.4,0.2);
    \draw[densely dotted] (0.6,0)node[below,scale=0.8]{$0.6$}--(0.6,0.3);
    
    \draw (1.07,0.57)node[red]{$\Phi_C$};
    \draw (1.07,0.47)node[blue]{$\Phi_D$};
    
    
    \draw[scale=1, domain=0:1, smooth, ultra thick, variable=\x, blue] plot ({\x}, {1/2 * \x});
    
    
    \draw[scale=1, domain=0:2/5, smooth, ultra thick, variable=\x, red] plot ({\x}, {1/4 * \x * \x + 3/10 * \x + 1/25});
    \draw[scale=1, domain=2/5:3/5, smooth, ultra thick, variable=\x, red] plot ({\x}, {1/2 * \x});
    \draw[scale=1, domain=3/5:1, smooth, ultra thick, variable=\x, red] plot ({\x}, {1/4 * \x * \x + 1/5 * \x + 9/100});
    
    \draw (1/2,-0.2)node{$\epsilon = 0.6$};
    \end{scope}

    \end{tikzpicture}
    \caption{Interim worst-case EU for risk neutral dividers and choosers with valuation $x$, facing $\cG(\mathcal{U}([0,1]),\epsilon)$ for $\epsilon \in \{0.02,0.4,0.6$\}.}
    \label{Figure: interim wc EU comparison}
\end{figure}
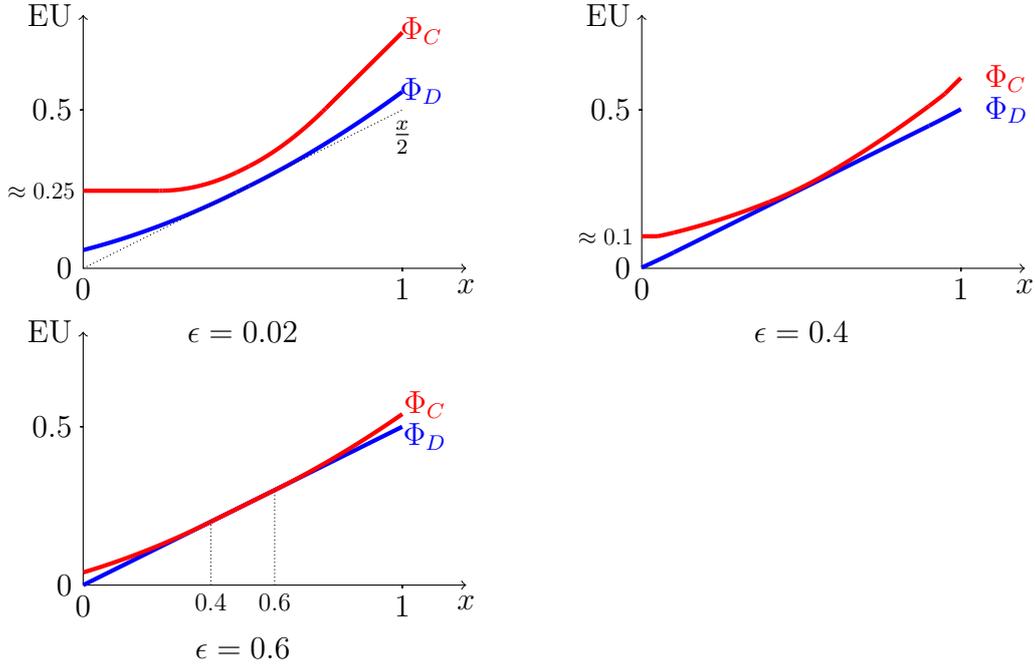

For $\epsilon < \tfrac{1}{4}$ we see that $\Phi_C(x) > \Phi_D(x)$ everywhere. In the case $\epsilon = 0.3>\tfrac{1}{4}$ we have $\Phi_C(x) = \Phi_D(x)$ precisely for $x \in [1-2\cdot 0.3, 2\cdot 0.3] = [0.4,0.6]$ as is clear from Theorem \ref{Theorem: interim EU comparison}. We also note that for increasing $\epsilon$ both functions are deforming into $x \mapsto \tfrac{x}{2}$, and the curvature of $\Phi_D$ is only visible for very small values of $\epsilon$.
\end{example}

The example combines and stresses two points of our argument for the Texas Shoot-Out under Knightian Uncertainty to being a deterrent exit mechanism. 
For small amounts of uncertainty, i.e. small $\epsilon$, the divider's interim worst-case expected payoff is close to the full-hedging payoff given by $\tfrac{x}{2}$. Thus, the Texas Shoot-Out lowers an agent's material incentive to initiate the mechanism and become a divider. We conclude that only co-owners with relatively high or low valuations might consider the exit profitable. Furthermore, the higher the uncertainty the more likely the mechanism is to achieve an efficient allocation. Secondly, even if a co-owner has a low/high valuation and is thus materially interested in ending the partnership under the terms of the Texas Shoot-Out, she would rather be in the position of the chooser and thus not triggering the mechanism in the first place.

\subsection{Correlation}\label{Section: correlation}
So far, we implicitly assumed that $\cG$ and $x_D$ are independent, implying that the divider's (uncertain) belief about the peer's valuation is the same for all her valuations. However, in the context of co-ownership one has good reason to assume that the assessment of the company's value is similar or more generally correlated among co-owners. Especially, the agents know more about the peer's valuation which could make us believe that the divider can make use of this additional information to extract profit. Interestingly however, we will see that suitable correlation will indeed achieve the exact opposite: The divider will not only play safe for more valuations but also do so for lower levels of uncertainty. It is thus not the mere lack of intelligence about the peer's valuation that fosters the efficiency of the Texas Shoot-Out in real world applications, but the interplay between uncertainty and the correlation of valuations.

To start the analysis, we can incorporate the idea of correlation (at an interim stage) by allowing for the lower and upper bounds $G_0, G_1$ to depend on $x_D$. In that case we write $G_0^{x_D}, G_1^{x_D}$. As always, the chooser is not affected by any uncertainty as she has all payoff relevant information after the price announcement. The divider takes into account for her price announcement while facing a seemingly more complicated problem. Fortunately, Theorem \ref{Theorem:divider characterization} is a point-wise statement and thus still remains valid, provided $G_i^{x_D}$ ($i=0,1$) fulfill the assumptions from Section \ref{Section: Optimal Price Announcement}. We only need to take into account the changes in $G_i^{x_D}$ and thereby $m_{\pi_{G_i^{x_D}}}$ with varying $x_D$.

For instance, a divider with valuation $x_D$ might expect the chooser's valuation to be drawn from a distribution $F^{x_D}$ with mode $x_D$ to emphasize that she considers $x_D$ to be the most likely valuation of the chooser. As before, she might want to account for a robust approach by means of a distribution band $\cG(F^{x_D},\epsilon)$. Note that the induced upper and lower bounds now depend on both $x_D$ and $\epsilon$.

The following example illustrates the impact of correlation by using triangular distributions.
\begin{example}\label{Example: correlation}
For $c \in [0,1]$ let $\Tri^c$ denote the triangular distribution with mode $c$ on $[0,1]$. Let the divider have valuation $x_D$. She might believe the chooser's valuation to be close to her own by considering a triangular distribution $\Tri^{x_D}$ with mode $c=x_D$. Accounting for uncertainty, she faces the distribution band $\cG(\Tri^{x_D},\epsilon)$. Now, applying Theorem \ref{Theorem:divider characterization} for $\epsilon = \tfrac{1}{5}$ we find that not matter $x_D$ the divider will always fully hedge herself, i.e. announce the price $\tfrac{x_D}{2}$. In comparison, Figure \ref{Figure: Triangular without correlation for c=1/2} illustrates the optimal price announcement for the case of a triangular distribution with a fixed and thus uncorrelated mode $c=\tfrac{1}{2}$.

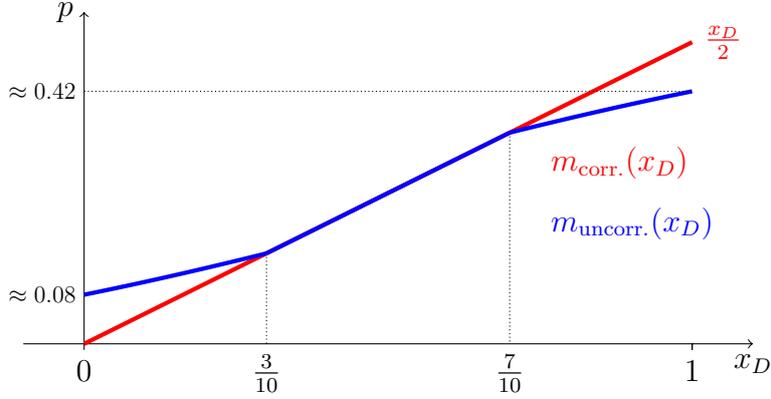
\begin{figure}
    \begin{center}
\begin{tikzpicture}[scale=8]
\draw[<->] (0,0.55)node[left]{$p$}--(0,0)--(1,0)--(1,-0.01)node[below]{$1$}--(1,0)--(1+1/10,0)node[below]{$x_D$};
\draw (-1/10,0)--(0,0);
\draw (0,0)--(0,-1/100)node[below]{$0$};

\draw[densely dotted](3/10,3/20)--(3/10,0)node[below]{$\tfrac{3}{10}$};
\draw[densely dotted](7/10,7/20)--(7/10,0)node[below]{$\tfrac{7}{10}$};

\draw[densely dotted](0,0)--(3/10,3/20);
\draw[densely dotted](7/10,7/20)--(1,1/2);
\draw[ultra thick,red](0,0)--(1,1/2)node[right]{$\tfrac{x_D}{2}$};
\draw[ultra thick, blue] (3/10,3/20)--(7/10,7/20);

\draw(0.75,0.3)node[right,red]{$m_{\text{corr.}}(x_D)$};
\draw(0.75,0.2)node[right,blue]{$m_{\text{uncorr.}}(x_D)$};

\draw[scale=1, domain=7/10:1, smooth, ultra thick, variable=\x, blue] plot ({\x}, {1/6*\x -1/60 * sqrt(100*\x*\x - 240*\x + 219)+2/5});
\draw[scale=1, domain=0:3/10, smooth, ultra thick, variable=\x, blue] plot ({\x}, {1/6*\x +1/60 * sqrt(100*\x*\x + 40*\x + 79)-1/15});

\draw(0,0.0814)node[left,scale=0.75]{$\approx 0.08$};
\draw[densely dotted](0,0.41853)node[left,scale=0.75]{$\approx 0.42$}--(1,0.41853);

\end{tikzpicture}
    \end{center}
    \caption{Optimal price announcement for $u = \id$ given the correlated (red) and uncorrelated (blue) distribution bands $\cG(\Tri^{x_D},\tfrac{1}{5})$ and $\cG(\Tri^{0.5}, \tfrac{1}{5})$ where $\Tri^c$ is the triangular distribution on $[0,1]$ with mode $c$.}
    \label{Figure: Triangular without correlation for c=1/2}
\end{figure}
\end{example}

Indeed, we can explicitly characterize the valuations $x_D$ for which a divider will play full-hedging in Example \ref{Example: correlation}.

\begin{proposition}\label{Proposition: correlation triangular}
Given the distribution band $\cG(\Tri^{x_D},\epsilon)$, a divider with valuation $x_D$ will play full-hedging if and only if
\begin{align*}
    \begin{cases}
    \sqrt{\frac{x_D}{2}}-\epsilon \leq x_D &, \text{if } x_D\leq \frac{1}{2}\\
    x_D \leq 1+\epsilon - \sqrt{\frac{1-x_D}{2}} &, \text{if } \frac{1}{2} \leq x_D
    \end{cases}.
\end{align*}
\end{proposition}

Example \ref{Example: correlation} illustrates that less uncertainty is needed for the divider to play safely if the chooser's valuation is believed to be correlated and close to her own valuation with a high probability. In fact, full-hedging for all valuations will already be played by the divider for $\epsilon \geq \tfrac{1}{8}$ in the correlated case whereas $\epsilon$ needs to exceed $\tfrac{1}{2}$ in the uncorrelated setting to generate this behavior. Intuitively, this observation arises as the medians of the stochastic dominated and stochastic dominant distribution of the distribution band lie around $x_D$ for many more valuations in the correlated case than in an uncorrelated one. In a nutshell, for the Texas Shoot-Out to generate an efficient outcome or prevent a premature exit less (while still some) uncertainty is needed if correlation between valuation seems reasonable.

\subsection{Full Uncertainty on Intervals}\label{Section: Uncertainty only about valuations}

So far, we modeled uncertainty as imprecise information about the distribution from which valuations are drawn. It might seem more natural to instead directly assume imprecise probabilistic information about the valuations, e.g. by believing a valuation to lie within an interval $[a,b] \subset [\xl,\xu]$. Our framework is flexible enough to accommodate this setting. By assuming $G_1$ to be the (discontinuous) distribution function that assigns point-mass $1$ to $a$ and $G_0$ to do so at $b$, we have $G(x) = 0$ if $x<a$, $G(x) = 1$ if $x>b$ and $0 \leq G(x) \leq 1$ if $x \in [a,b]$ for all $G \in \cG$. Especially, $\mu_{G_1} = a, \mu_{G_0} = b$. Note that $\pi_{G_0}(. \mid x_D), \pi_{G_1}(. \mid x_D)$ are both strictly quasiconcave with $m_{\pi_{G_0}}(x_D) = b, m_{\pi_{G_1}}(x_D) = a$ while generically not being continuous in that respective point. Even more, $\pi_{G_0}(. \mid x_D)$ resp. $\pi_{G_{1}}(. \mid x_D)$ attain their maximum if and only if $b \leq x_D$ resp. $a \leq x_D$. However, Theorem \ref{Theorem: interim EU comparison} still reveals 'the' optimal price announcement:

The divider will announce a price $\tfrac{x_D}{2}$ if $x_D \in [a,b]$ and offer $\tfrac{b}{2}$ if $x_D>b$. For $x_D < a$, the supremum of the divider's payoff function is $\lim_{p \nearrow a/2} \pi(p \mid x_D)$, which cannot be attained. Practically however, she can announce a lower price as close to $\tfrac{a}{2}$ as possible, while still giving the divider an incentive to sell. Note that this issue is related to assuming the chooser's strategy to be 'sell' if and only if $x_C-p \leq p$. If we instead take the strict inequality, the problem shifts towards valuations above the median. 

Intuitively, the divider faces full uncertainty about $x_C$ belonging to $[a,b]$ and thus plays full-hedging if $2p \in [a,b]$. If the divider has a valuation $x_D \geq b$, she is certain to be more interested in buying the company and will ensure this outcome by offering the lowest price for which the chooser is selling her shares for sure, which is $b/2$. Vice versa, for a valuation $x_D < a$, she wants to keep the price as high as possible, while still ensuring the chooser to buy her shares. She can do so by announcing a price (infinitesimal below) $a/2$. Figure \ref{Figure: Optimal price announcement full uncertainty on interval} depicts this observation.

\begin{figure}
    \centering
    \begin{tikzpicture}[scale=7]
    \draw[<->] (0,0.45)--(0,0)node[below]{$\xl$}--(1,0)--(1,-0.01)node[below]{$\xu$}--(1,0)--(1.1,0)node[below]{$x_D$};
    \draw[densely dotted](0.2,0)node[below]{$a$}--(0.2,0.1);
    \draw[densely dotted](0.7,0)node[below]{$b$}--(0.7,0.35);
    \draw[densely dotted](0,0.35)node[left]{$\frac{b}{2}$}--(0.7,0.35);
    \draw[blue,dashed](0,0.1)node[left,black]{$\frac{a}{2}$}--(0.2,0.1);
    \draw[blue](0.2,0.1)--(0.7,0.35)--(1,0.35);
    \draw[blue](0.9,0.3)node[]{$m(x_D)$};
    \end{tikzpicture}
    \caption{Optimal price announcement for full uncertainty on $[a,b]$. Between $\xl$ and $a$ the supremum of the worst-case EU cannot be attained, but approximated for $p$ close to, but below $\tfrac{a}{2}$.}
    \label{Figure: Optimal price announcement full uncertainty on interval}
\end{figure}
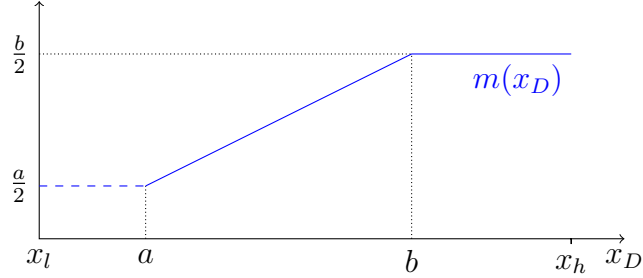

More generally, we can derive an 'approximate' optimal price announcement by invoking the interplay of results from the Appendix (Proposition \ref{Proposition: stict quasi concave minimum mh} and Lemma \ref{Lemma: median and optimal price}) with limited assumptions on the regularity of the distribution functions as long as they induce strictly quasiconcave utility functions. The main analytic issue is just whether the supremum can be attained. By announcing a price that achieves a worst-case expected payoff arbitrarily close to the supremum value, Theorem \ref{Theorem:divider characterization} keeps its simple structure and can be applied in real world scenarios.

\section{Conclusion}\label{Section: Conclusion}
When founding a co-owned company it is preventative and advantageous to  agree in advance on a dissolution mechanism in case things get sour. Any such dissolution mechanism should be readily available, simple,  and it should  discourage a selfish premature termination of the partnership. The so-called Texas Shoot-Out is a well-known example of such an exit mechanism for two co-owners. The co-owner initiating the mechanism announces a price for the sole ownership of the company while the other can choose to sell or buy the company at that price. While simple and independent of external outside options, it is notorious for its deterrent effect on a premature dissolving.

This article invokes Knightian Uncertainty as an explanatory source for this discouragement. Having in mind that co-owners have some idea about the distribution of their partner, we allow for any degree of confidence in a reference distribution $F$ by modeling the uncertainty as an distribution band around $F$ that can range anywhere in between the Bayesian single prior setting towards one of full uncertainty. If the induced worst-case expected utility functions are strictly quasiconcave, we derive the co-owners' optimal actions and interim expected payoffs. Our main result is the explicit characterization of the the divider's price announcement which is a surprising mixture of the optimal strategies under no and full uncertainty: She will play cautiously for valuations close to the median valuation of the reference distribution while still trying to generate a revenue that exceeds half her valuation for low or high valuations. Hence, only co-owners with low/high valuations have a material incentive to initiate a Texas Shoot-Out. However, these co-owners strictly prefer the other co-owner to trigger the exit mechanism as they otherwise find themselves in an unfavorable position (being forced to leave the company resp. take over the sole ownership). Welfare is improved as uncertainty increases efficiency, i.e., for more valuation profiles the co-owner with higher valuation becomes the sole owner. These desirable consequences are already visible for small levels of uncertainty and even enhanced if there is good reason to believe the valuations to be positively correlated. Knightian Uncertainty can thus explain why consultancies recommend to include the Texas Shoot-Out in buy-sell agreements.

\newpage
{\footnotesize
\bibliographystyle{econometrica}
\bibliography{literature}
}
\newpage

\appendix
\setcounter{cor}{0}
\renewcommand{\thecor}{\Alph{section}\arabic{cor}}
\setcounter{lemma}{0}
\renewcommand{\thelemma}{\Alph{section}\arabic{lemma}}
\setcounter{proposition}{0}
\renewcommand{\theproposition}{\Alph{section}\arabic{proposition}}
\setcounter{definition}{0}
\renewcommand{\thedefinition}{\Alph{section}\arabic{definition}}

\section{Quasiconcavity}\label{Appendix: Quasiconcavity}
In the following, we give formal definitions of the concept of quasiconcave functions and some of their properties relevant to our analysis of the Texas Shoot-Out.

\begin{definition}[Quasiconcavity]
Let $S \subseteq \R^L$ be a convex set. A function $f \colon S \to \R$ is called \emph{quasiconcave} if one of the following equivalent statements holds true:
\begin{enumerate}[(i)]
    \item For any $x,x' \in S$ and any $\lambda \in (0,1)$ we have $f(\lambda x + (1-\lambda) x') \geq \min \{ f(x),f(x') \}$.\label{def: qc min}
    \item Every super level set is convex, i.e. for any $\xi \in \R$ the set $\{ s \in S \, | \, f(s) \geq \xi \}$ is convex.\label{def: qc superlevel}
\end{enumerate}
If $L = 1$ and $S$ is an interval, there is another equivalent statement that describes the shape of a quasiconcave function. For ease of exposition we stick to a compact interval $S = [a,b]$ in the following. However, the given proof works as well for half-open or open ones by just replacing the respective parentheses accordingly.
\begin{enumerate}[(i)]
    \setcounter{enumi}{2}
    \item There exists a number $m \in [a,b]$ such that (at least) one of the following properties hold:
    \begin{enumerate}[(a)]
        \item $f \res_{[a,m]}$ is increasing and $f \res_{(m,b]}$ is decreasing.\label{def: mon left closed}
        \item $f \res_{[a,m)}$ is increasing and $f \res_{[m,b]}$ is decreasing.\label{def: mon right closed}
    \end{enumerate}\label{def: qc monotone}
\end{enumerate}

It is worth noting in the above cases that both, the restriction of $f$ to the empty set as well as to a single point are considered (strictly) increasing as well as (strictly) decreasing\footnote{E.g. a function $f \colon X \to Y$ is increasing by definition if $\forall x,y \in X (x \leq y \Rightarrow f(x) \leq f(y))$ which is true if $X$ is a singleton or empty. A consequence of this is that a non-increasing function has at least two elements in its domain.}. Especially, if $m=b$ in \eqref{def: mon left closed} holds, this simply means that $f$ is increasing. Similarly, $m=a$ in \eqref{def: mon right closed} corresponds to $f$ being decreasing.

\begin{proof}
We first prove the more general equivalence of \eqref{def: qc min} and \eqref{def: qc superlevel} and will then show \eqref{def: qc min}$\Leftrightarrow$\eqref{def: qc monotone}.
\begin{itemize}
    \item \eqref{def: qc min}$\Rightarrow$\eqref{def: qc superlevel}:\\
    Let $\xi \in \R$ and let $x,x' \in \{ s \in S \, | \, f(s) \geq \xi \}$. Consider any convex combination $x'' := \lambda x + (1-\lambda)x'$. By \eqref{def: qc min} $f(x'') \geq \min \{ f(x),f(x') \} \geq \xi$.
    
    \item \eqref{def: qc superlevel}$\Rightarrow$\eqref{def: qc min}:\\
    Define $\xi := \min\{ f(x), f(x') \}$. Then $\{ s \in S \, | \, f(s) \geq \xi \}$ is convex and contains $x,x'$. Thus, it also contains $\lambda x + (1-\lambda)x'$ for any $\lambda \in (0,1)$ and by definition $f(\lambda x + (1-\lambda)x') \geq \xi = \min \{ f(x), f(x') \}$.
    
    \item \eqref{def: qc monotone}$\Rightarrow$\eqref{def: qc min}:\\
    If $m \in \{a,b\}$, $f$ surely fulfills \eqref{def: qc min}. For $m \in (a,b)$ consider now an $f$ that has property \eqref{def: mon left closed} or \eqref{def: mon right closed}. Let $x,x' \in [a,b]$, wlog $x<x'$ and $\lambda \in (0,1)$. 
    If $\widetilde{x} := \lambda x + (1-\lambda) x' < m$, we clearly have $f(\widetilde{x}) \geq f(x) \geq \min\{ f(x),f(x') \}$ since $f$ is increasing on $[a,\widetilde{x}] \subseteq [a,m)$. If $m < \widetilde{x}$, we find $f(\widetilde{x}) \geq f(x') \geq \min \{ f(x),f(x') \}$ as $f$ is decreasing on $[\widetilde{x},b] \subseteq (m,b]$. Finally, if $\widetilde{x} = m$ we find $f(\widetilde{x}) \geq f(x) \geq \min \{f(x),f(x')\}$ if condition \eqref{def: mon left closed} holds and $f(\widetilde{x}) \geq f(x') \geq \min \{f(x),f(x')\}$ if condition \eqref{def: mon right closed} holds. Thus, \eqref{def: qc min} is true.
    \item  \eqref{def: qc min}$\Rightarrow$\eqref{def: qc monotone}:\\
    By means of contraposition: Assume the negation of \eqref{def: qc monotone} which implies that for each $m \in [a,b]$ at a time the restriction $f \res_{[a,m)}$ is not increasing or $f \res_{(m,b]}$ is not decreasing (or both).
    
    First suppose that there is an $m$ such that both $f \res_{[a,m)}$ is not increasing and $f \res_{(m,b]}$ is not decreasing. This implies the existence of $x<x' \leq m \leq x'' < x'''$ with $f(x)>f(x')$ and $f(x'')<f(x''')$. Especially, $\min\{ f(x'),f(x'') \} < \min \{ f(x),f(x''') \}$, showing, that \eqref{def: qc min} does not hold. A similar argument is valid for the negation of \eqref{def: mon right closed}.
    
    Second, consider the case where for each $m$ at a time either $f \res_{[a,m)}$ is not increasing or $f \res_{(m,b]}$ is not decreasing but \emph{never both}. Equivalently, for each $m$ at a time $f \res_{[a,m)}$ is increasing or $f \res_{(m,b]}$ is decreasing but never both.
    The set $M_i := \{ m \in [a,b] \, | \, f \res_{[a,m)} \text{ is increasing} \}$ is non-empty since formally $a \in M_i$. Thus, $m_i := \sup M_i$ is well-defined. Note that $a<m_i$ since $a = m_i$ implies that $f$ obeys \eqref{def: mon left closed} which was excluded. Likewise, $m_i<b$, since otherwise $f$ would fulfill \eqref{def: mon right closed}. Thus, for the infimum $m_d := \inf M_d$ for $M_d := \{ m \in [a,b] \, | \, f\res_{(m,b]} \text{ is decreasing} \}$ is well-defined. In the considered case and by their definitions we must have $m_i \leq m_d$ and even $m:= m_i = m_d \in (a,b)$. This means that $f\res_{[a,m)}$ is increasing and $f \res_{(m,b]}$ is decreasing. Since neither \eqref{def: mon left closed} nor \eqref{def: mon right closed} are true by assumption, $\sup_{x < m}f(x) > f(m)$ and $f(m) < \inf_{m < x}f(x)$. Thus there are $x < m < x'$ such that $f(m) < \min \{f(x),f(x')\}$, showing that \eqref{def: qc min} does not hold.
    \end{itemize}
\end{proof}
\end{definition}

\begin{definition}[Strict quasiconcavity] \label{Definition: Strict quasiconcavity}
Let $S \subseteq \R^L$ be a convex set. A function $f \colon S \to \R$ is called \emph{strictly quasiconcave} if
\begin{enumerate}[(i')]
    \item For any $x,x' \in S, x < x'$ and any $\lambda \in (0,1)$ we have $f(\lambda x + (1-\lambda) x') > \min \{ f(x),f(x') \}$.\label{def: strict qc min}
\end{enumerate}
If $L = 1$ and $S$ is an interval, the following assertion (stated for a compact interval $S=[a,b]$) is an additional characterization of strict quasiconcavity.
\begin{enumerate}[(i')]
    \setcounter{enumi}{2}
    \item There exists a unique number $m \in [a,b]$ such that (at least) one of the following properties hold:
    \begin{enumerate}[(a)]
        \item $f \res_{[a,m]}$ is strictly increasing and $f \res_{(m,b]}$ is strictly decreasing.\label{def: strict mon left closed}
        \item $f \res_{[a,m)}$ is strictly increasing and $f \res_{[m,b]}$ is strictly decreasing.\label{def: strict mon right closed}
    \end{enumerate}\label{def: strict qc monotone}
\end{enumerate}
\begin{proof}
The proof is an straightforward adjustment of the one for quasiconcavity. The uniqueness of $m$ follows as strictly quasiconcave functions can't have horizontal sections.
\end{proof}
\end{definition}

\begin{definition}
If $f$ is strictly quasiconcave and $S$ an interval, we denote the unique value that separates the increasing part of the function from the decreasing one by $m_f$.
\end{definition}
Note that if $f$ attains a global maximum (e.g. if $f$ is continuous), it must be $m_f$. If $f$ does not have a global maximum, it's supremum value is either $\lim_{x \nearrow m_f} f(x)$ or $\lim_{x \searrow m_f} f(x)$.

Obviously, strictly quasiconcave functions are quasiconcave. Also note that (strictly) concave functions are (strictly) quasiconcave, but quasiconcave functions need not be concave nor even continuous.

The following lemma shows some handy properties the minimum of two strictly quasiconcave functions.
\begin{lemma}\label{Lemma: h quasiconcave}
Let $f,g \colon S \to \R$ be strictly quasiconcave. Then the point-wise minimum $h := \min\{f,g\}$ is strictly quasiconcave as well.
Furthermore, if $S=[a,b]$ and, e.g., $m_f \leq m_g$, then $m_h = \inf \, \{ x | m_f \leq x < m_g, f(x) \leq g(x) \}\cup \{ m_g\}$.
\begin{proof}
Let $x<x'$ and $\lambda\in (0,1)$. Then
\begin{align*}
    h(\lambda x + (1-\lambda)x') &= \min \{ f(\lambda x + (1-\lambda)x'), g(\lambda x + (1-\lambda)x') \}\\
    &> \min \{ \min \{f(x),f(x')\}, \min \{ g(x), g(x') \} \}\\
    &= \min \{ \min \{ f(x),g(x) \}, \min \{ f(x'),g(x')  \}  \}\\
    &= \min \{ h(x),h(x') \}.
\end{align*}
We note that, replacing the strict inequality for a weak one, the minimum of two quasiconcave functions is also quasiconcave.

We now turn towards the points determining increasing and decreasing domains. By definition, $h$ is strictly increasing on $[a,m_f)$ and strictly decreasing on $(m_g,b]$, thus $m_h \in [m_f,m_g]$. If $m_f=m_g$, $m_h = m_f=m_g$. We thus assume $m_f < m_g$ in the following. Consider the set $M := \{ x \, | \, m_f \leq x < m_g, f(x) \leq g(x) \}$. If $M$ is empty, we have $h\res_{[m_f,m_g)} \equiv g\res_{[m_f,m_g)}$, which is increasing, and thus necessarily $m_h = m_g$. If $M$ is non-empty, let $x_0 := \inf M$ and note $x_0 < m_g$. We have $m_h \leq x_0$ as $h \res_{(x_0,m_g)} = f \res_{(x_0,m_g)}$ keeping in mind that $f$ is strictly decreasing on that domain. By definition, $m_f \leq x_0$. If $m_f = x_0$, we necessarily have $m_h = x_0 = m_f$. If $m_f < x_0$, the interval $[m_f,x_0)$ is non-empty and we find $h \res_{[m_f,x_0)} = g \res_{[m_f,x_0)}$, which is strictly increasing on the considered domain. Thus also $x_0 \leq m_h$ and hence $x_0 = m_h$.
\end{proof}
\end{lemma}

Note that in a setting with continuous strictly quasiconcave functions, the proposition simplifies: The point $m_h$ is the minimum among $m_f,m_g$ and (if it exists) the unique intersection point between these two.

For our purposes, the following special case is of good use for our main theorem:
\begin{proposition}\label{Proposition: stict quasi concave minimum mh}
Let $f,g$ be strictly quasiconcave and assume that $h := \min \{ f,g\}$ has the following form for some $x_0$
\begin{align*}
    h(x) = \begin{cases}
    g(x) &, x<x_0,\\
    f(x) = g(x) &, x=x_0,\\
    f(x) &, x_0<x
    \end{cases}.
\end{align*}
Then,
\begin{align*}
    m_h = \begin{cases}
    m_f &, x_0 < m_f,\\
    x_0 &, m_f \leq x_0 \leq m_g,\\
    m_g &, m_g < x_0
    \end{cases}.
\end{align*}
\begin{proof}
We first note that by Lemma \ref{Lemma: h quasiconcave} $h$ is strictly quasiconcave.
Now observe that the inequality $m_g < x_0 < m_f$ is impossible: Otherwise, $h = g$ is strictly decreasing on $(m_g,x_0)$ while $h = f$ is strictly increasing on $(x_0,m_f)$ - a contradiction to $h$ being strictly quasiconcave.

Hence, it remains to investigate the cases where $x_0 \leq m_g$ or $m_f \leq x_0$ (or both) which defines three distinct cases:
\begin{enumerate}[(i)]
    \item If both are true, i.e. $m_f \leq x_0 \leq m_g$, we find that $h=f$ is strictly increasing on $(a,x_0)$\footnote{If $a=x_0$, the interval $(a,x_0)$ is empty, but the statement is still true. The same argument applies for $x_0=b$.} while $h=g$ is strictly decreasing on $(x_0,b)$. Thus, $m_h = x_0$.
    \item If $x_0 \leq m_g$ and $x_0 < m_f$, we find that $h=f$ is strictly increasing on $(x_0,m_f)$ while strictly decreasing on $(m_f,b)$. Thus, $m_h=m_f$.
    \item If $m_f \leq x_0$ and $m_g < x_0$, we find that $h=g$ is strictly increasing on $(a,m_g)$ while strictly decreasing on $(m_g,x_0)$. Thus, $m_h = m_g$.
\end{enumerate}
\end{proof}
\end{proposition}

\section{Proofs}
\setcounter{cor}{0}
\renewcommand{\thecor}{\Alph{section}\arabic{cor}}
\setcounter{proposition}{0}
\setcounter{lemma}{0}
\renewcommand{\thelemma}{\Alph{section}\arabic{lemma}}
\setcounter{proposition}{0}
\renewcommand{\theproposition}{\Alph{section}\arabic{proposition}}
\setcounter{definition}{0}
\renewcommand{\thedefinition}{\Alph{section}\arabic{definition}}

\begin{proof}[Proof of Lemma \ref{Lemma: Suff McAfee}]
The proof essentially follows the lines of Lemma 7 in \cite{mcafee92}. It is worth mentioning that there's a typo at the second condition in the reference which is clear from their proof. Also note that the 'price' in \cite{mcafee92} refers to the claimed overall price for the indivisible object, while we refer to it as the value of one's or the other's share which is half of that.

We start by calculating
\begin{align*}
    &\frac{\partial}{\partial p} \pi_F(p \mid x_D)\\
    =& f(2p) \cdot \left\{ 2u(x_D-p) -2u(p)  + \frac{F(2p)}{f(2p)} \cdot \left( -u'(x_D-p) - u'(p) \right) + \frac{u'(p)}{f(2p)}\right\}.
\end{align*}
Now, the second derivative\footnote{In fact, for the proof it suffices to have second left and right derivatives of $F$ resp. fulfilling the SHRC as long as the first derivative is continuous.} of $\pi_F$ with respect to $p$ and evaluated in a point $q$ such that $2q \in [\xl,\xu]$ (if it exists) with $\tfrac{\partial}{\partial p}\pi_F(p \mid x_D) \mid_{p=q} = 0$ is
\begin{align*}
    \frac{\partial^2}{\partial p^2} \pi_F(p \mid x_D) \mid_{p=q} =& - f(2q) \frac{\partial}{\partial p} \left( 2p + \frac{F(2p)}{f(2p)} \right) \mid_{p=q} \cdot u'(x_D-q)\\
    &- f(2q) \frac{\partial}{\partial p} \left( 2p- \frac{1-F(2p)}{f(2p)}  \right) \mid_{p=q} \cdot u'(q)\\
    &+ u''(x_D-q) F(2q) + u''(q) (1-F(2q)) < 0.
\end{align*}
Thus, any $q$ with $\tfrac{\partial}{\partial p} \pi_F(p \mid x_D) \mid_{p=q} = 0$ is an isolated local maximum. There can't be more than one such isolated local maximum $q$, since if there were (at least) two, there must be infinitely many in between them: Let $q_1<q_2$ with $\tfrac{\partial}{\partial p} \pi_F (p \mid x_D) \res_{p=q_i} = 0$, then by the above calculation $\tfrac{\partial^2}{\partial p^2} \pi_F(p \mid x_D) \res_{p=q_i} <0$, thus there exists $q_1 <p_1 < p_2 < q_2$ with $\tfrac{\partial}{\partial p} \pi_F (p \mid x_D) \res_{p=q_1} <0, \tfrac{\partial}{\partial p} \pi_F (p \mid x_D) \res_{p=q_2} >0$. By continuity of $\tfrac{\partial}{\partial p} \pi_F (p \mid x_D)$ there must thus be a point $m_1 < q_3 < p_2$ with $\tfrac{\partial}{\partial p} \pi_F (p \mid x_D) \res_{p=q_3} = 0$. That way, we can construct infinitely many isolated maxima in $[q_1,q_2]$. An accumulation point of these must be a zero of the first derivative of $\pi_F$ itself be continuity of $\tfrac{\partial}{\partial p} \pi_F$. But this can't be an isolated maximum - a contradiction.

Now, note that if there's a unique such $q$, then $\tfrac{\partial}{\partial p} \pi_F(p \mid x_D)$ must be positive before $q$ and negative afterwards by continuity. Thus, it's strictly quasiconcave.

Finally, if there's no such $q$, $\tfrac{\partial}{\partial p} \pi_F$ must either be positive or negative on the whole domain and thus strict quasiconcavity of $\pi_F(p \mid x_D)$ holds in both cases.
\end{proof}

\begin{proof}[Proof of Lemma \ref{Lemma: Priceinterval Bayes}]
Keep in mind that the divider can always ensure a payoff of $u(\tfrac{x_D}{2})$ by the full-hedging strategy $\bar{p} = \tfrac{x_D}{2}$. An optimal price should not be worse than this. Now, firstly consider a price offer of $2p<\xl$. As $x_C-p >p$, the chooser will take the offer and the divider faces a payoff of $u(p) < u(\tfrac{x_D}{2})$. Secondly, if a price offer with $\xu < 2p$ is made, the chooser will sell her shares, yielding a payoff of $u(x_D - p) < u(\tfrac{x_D}{2})$ for the divider. Consequently, the divider only considers prices $2p \in [\xl,\xu]$.
\end{proof}
It is worth noting that the proof of Lemma \ref{Lemma: Priceinterval Bayes} only makes use of $F$ being a distribution function, i.e. being $0$ resp. $1$ for $x<\xl$ resp. $\xu < x$, leading to the functions $u(p)$ resp. $u(x_D - p)$.\\
While we're at it, the same thought reveals that $\pi_G$ being strictly quasi concave in $p$ for $2p \in X$ (Assumption \ref{AssF}) is equivalent to strict quasiconcavity everywhere, c.f. the end of the proof of Lemma \ref{Lemma: suff epsilon}.

\begin{proof}[Proof of Proposition \ref{Proposition: McAfee}]
The proof is adapted from \cite{mcafee92}'s Lemma 7, incorporating the more general set-up using strict quasi-concavity as the main driver instead of the SHRCs.

The chooser's best response has already been discussed, thus we focus on the divider's best reply.

Note that we can restrict to an interim perspective here as we simply maximize the integrand of the expectation for each $x_D$ at a time and there's no correlation between the two valuations. Since $\pi_F(p \mid x_D)$ is strictly quasiconcave and continuous in $p$, there is a unique price announcement $m(x_D)$ for each $x_D$ maximizing the function, see Appendix \ref{Appendix: Quasiconcavity}.

Firstly, note that $\tfrac{\partial}{\partial p} \pi_F(p\mid x_D) \res_{2p = \xl} = f(\xl) \cdot ( 2 \cdot (u(x_D - \tfrac{\xl}{2}) - u(\tfrac{\xl}{2}))) + u'(\tfrac{\xl}{2}) > 0$ since $F(\xl) = 0$, $x_D \geq \xl$ and $u$ is strictly increasing with $u' >0$. Similarly, we find $\tfrac{\partial}{\partial p} \pi_F(p\mid x_D) \res_{2p = \xu} < 0$. Together with continuously differentiability, this implies that $m(x_D)$ is an interior local maximum, i.e. $\xl < 2m(x_D) < \xu$ and $\tfrac{\partial}{\partial p} \pi_F(p \mid x_D) \res_{p = m(x_D)} = 0$.

Secondly, as a zero of the derivative function, it suffices to show that $\tfrac{\partial^2}{\partial x_D \partial p} \pi_F(p \mid x_D) > 0$ to prove that $m(x_D)$ is strictly increasing. This is immediate from
\begin{equation*}
    \frac{\partial^2}{\partial x_D \partial p} \pi_F(p \mid x_D) = 2 f(2p) \cdot u'(x_D-p)-F(2p)\cdot u''(x-p) > 0,
\end{equation*}
since $f>0$ and $u$ is strictly increasing and concave with $u'>0$.

Recall that $\mu_F$ is the median of $F$, i.e. $F(\mu_F) = \tfrac{1}{2}$. Calculating $\tfrac{\partial}{\partial p} \pi_F(p \mid x_D) \res_{2p = x_D} = u'(\tfrac{x_D}{2}) \cdot (1 - 2F(x_D))$ and checking its sign, we find that $x_D \lessgtr 2m(x_D)$ if $x_D \lessgtr \mu_F$. In addition, since $\tfrac{\partial^2}{\partial p^2} \pi_F(p \mid x_D) \res_{2p = x_D} = -8f(x_D) u'(\tfrac{x_D}{2}) + u''(\tfrac{x_D}{2}) < 0$ the function $\pi_F(p \mid x_D)$ is locally strictly concave in $2p = x_D$, thus, $x_D = \mu_F$ implies $2m(\mu_F) = \mu_F$.

Finally, we turn to the relation between $2m(x_D)$ and $\mu_F$. Consider, e.g., $x_D < \mu_F$. We already know that this implies $x_D < 2m(x_D)$. By the properties of $u$ we thus have $u(x_D - m(x_D)) < u(m(x_D))$ and $u'(x_D - m(x_D)) > u'(m(x_D))$. This implies
\begin{align*}
    0 &= \frac{\partial}{\partial p} \pi_F(p \mid x_D) \res_{p = m(x_D)}\\
    &= -u'(x_D - m(x_D)) F(2m(x_D)) + u'(m(x_D)) (1-F(2m(x_D))\\
    &\phantom{==}+ 2f(2m(x_D))(u(x-m(x_D)) - u(m(x_D)))\\
    &< u'(m(x_D)) \cdot (1 - 2F(2m(x_D))),
\end{align*}
revealing $2m(x_D)< \mu_F$. A similar argument shows that $x_D > \mu_F$ implies $2m(x_D)>\mu_F$.\\

Note that the chooser's decision rule, accepting the offer if and only if $x_C-p \leq p$ for all $p$ and $x_C$, is always a best reply. Since the divider's price announcement is a unique best reply for each $x_D$ they form an interim Bayesian Nash equilibrium and thus a Bayesian Nash equilibrium in the ex-ante stage as well.\\

The statement about interim expected utility is proved exactly as Theorem 9 in \cite{mcafee92}. We will investigate a generalization of this using Knightian Uncertainty in Section \ref{Section: interim utility}.
\end{proof}

\begin{proof}[Explicit calculations to Example \ref{Example: uniform distribution}]
It is straightforward to show that the induced respective utility functions are given by
\begin{align*}
\pi_{G_0}(x_D,p) &= \begin{cases}
p &, 0 \leq 2p \leq 2\epsilon,\\
-4p^2 + 4p\epsilon + x_D(2p-2\epsilon) + p &, 2\epsilon \leq x < 1,\\
x_D-p &, 2p = 1.
\end{cases}\\
\pi_{G_1}(x_D,p) &= \begin{cases}
-4p^2 - 4p\epsilon + x_D(2p+2\epsilon) + p &, 0\leq 2p \leq 1- 2\epsilon\\
x_D-p &, 1-2\epsilon < 2p \leq 1.
\end{cases}.
\end{align*}
Note that $\pi$ is the function that stays $\pi_{G_0}$ until $2p=x_D$ and is $\pi_{G_1}$ afterwards. Moreover, on $2p \in [\xl,\xu)$ both, $\pi_{G_0}$ and $\pi_{G_1}$ are strictly quasiconcave.

The unique local maxima of the parabolas are located at $2m_0(x_D) = \tfrac{x_D}{2} + \epsilon + \tfrac{1}{4}$ and $2m_1(x_D) = \tfrac{x_D}{2} - \epsilon + \tfrac{1}{4}$, respectively. Note that $m_1(x_D) < m_0(x_D)$ and moreover $2m_1(x_D) \leq x_D \iff \tfrac{1}{2} - 2\epsilon \leq x_D$ as well as $x_D \leq 2m_0(x_D) \iff x_D \leq \tfrac{1}{2} + 2\epsilon$. The following conclusions arise:

Firstly, for $x_D \in [\tfrac{1}{2}-2\epsilon, \tfrac{1}{2} + 2 \epsilon]$ the highest value of the function $\pi$ is attained at $2p=x_D$ since it is strictly increasing for values below it and strictly decreasing for higher values of $2p$.

Secondly, for $x_D < \tfrac{1}{2}-2\epsilon$ we find $x_D < 2m_1(x_D) \leq 1$ and the maximum of $\pi$ is attained at $p=m_1(x_D)$ ($\pi_{G_0}$ is increasing, switching into $\pi_{G_1}$ and still strictly increasing until $p=m_1(x_D)$ and strictly decreasing afterwards).

Finally, for $\tfrac{1}{2}+2\epsilon < x_D$ we find $0 \leq 2m_0(x_D) < x_D$ an the maximum of $\pi$ is attained at $p=m_0(x_D)$ ($\pi_{G_0}$ strictly increases, reaches its top and starts to strictly decrease and keeps doing to after switching to $\pi_{G_1}$).
\end{proof}

\begin{proof}[Proof of Lemma \ref{Lemma: suff epsilon}]
We will indeed show, that the corresponding payoff functions are strictly quasi-concave not only on $X$ but on $\R$. We first note that the stochastic dominant/dominated distributions in $\cG(F,\epsilon)$ are given by
\begin{align*}
    G_0(x) &= \begin{cases}
    0 &, F(x-\epsilon) <0,\\
    F(x-\epsilon), &0 \leq F(x-\epsilon), x < \xu,\\
    1 &, \xu \leq x
    \end{cases}
    \\
    G_1(x) &= \begin{cases}
    0 &, x<\xl,\\
    F(x+\epsilon) &, \xl \leq x, F(x+\epsilon)<1,\\
    1 &, 1 \leq F(x+\epsilon).
    \end{cases}
\end{align*}
We first look on the resp. domains where the functional form of the distribution is given by $F(x\pm \epsilon)$. Since $F$ fulfills the SHRC (Lemma \ref{Lemma: Suff McAfee}) so do $G_0,G_1$ as
\begin{align*}
    \frac{\partial}{\partial x} \left( x + \frac{F(x \pm \epsilon)}{f(x \pm \epsilon)} \right) = \frac{\partial}{\partial x} \left( x \pm \epsilon + \frac{F(x \pm \epsilon)}{f(x \pm \epsilon)} \right) \geq 0,\\
    \frac{\partial}{\partial x} \left( x - \frac{1-F(x \pm \epsilon)}{f(x \pm \epsilon)} \right) = \frac{\partial}{\partial x} \left( x \pm \epsilon - \frac{1-F(x \pm \epsilon)}{f(x \pm \epsilon)} \right) \geq 0,
\end{align*}
by our assumption on $F$. Note that the induced payoff functions on the domains where a CDF is equal to $0$ resp. $1$ are $u(p)$ resp. $u(x_D-p)$ which are strictly increasing resp. decreasing. Thus, to ensure global strictly quasiconcavity, it we thus only need to look at the (possible) discontinuities as they could produce jumps that are not in line with Definition \ref{Definition: Strict quasiconcavity}. Specifically, it suffices to show $\lim_{2p \nearrow \xu} \pi_{G_0}(p\mid x_D) \geq \pi_{G_0}(\tfrac{\xu}{2}\mid x_D)$ and $\lim_{2p \nearrow \xl} \pi_{G_1}(p \mid x_D) \leq \pi_{G_1}(\tfrac{\xl}{2} \mid x_D)$.

Firstly, we have $\lim_{2p \nearrow \xu} \pi_{G_0}(p \mid x_D) = u(x_D-\frac{\xu}{2}) F(\xu-\epsilon) + u(\frac{\xu}{2})(1-F(\xu-\epsilon)) \geq  u(x_D - \frac{\xu}{2})=\pi_{G_0}(\frac{\xu}{2})$. Thus, $\pi_{G_0}$ is strictly quasiconcave on $\R$.

Secondly, we have $\lim_{2p \nearrow \xl} \pi_{G_1}(p \mid x_D) = u(\frac{\xl}{2}) \leq u(x_D - \frac{x}{2}) F(\xl+\epsilon) + u(\frac{\xl}{2})(1 - F(\xl + \epsilon)) = \pi_{G_1}(\frac{\xl}{2} \mid x_D)$. Thus, $\pi_{G_1}$ is strictly quasiconcave on $\R$.
\end{proof}

\begin{proof}[Calculations for Example \ref{Example: SHRC}]
It suffices to check the SHRC on the distribution's support since all of them are continuous on $\R$ (compare this to the 'glueing' argument at the end of the proof of Lemma \ref{Lemma: suff epsilon}).
\phantom{a}\\
{\bfseries Triangular distribution:}\\
On an interval $[a,b]$, the triangular distribution with mode $c \in [a,b]$ is defined by
\begin{equation*}
    F(x) = \Tri^c(x) = \begin{cases}
    0 &, x<a,\\
    \frac{(x-a)^2}{(b-a)(c-a)} &, a \leq x <c,\\
    \frac{c-a}{b-a} & x=c,\\
    1 - \frac{(b-x)^2}{(b-a)(b-c)} &, c < x \leq b,\\
    1 &, b<x.
    \end{cases}
\end{equation*}
Note that $x=c$ can be seen as a limit case and that its density $f$ is continuous on $[a,b]$. The proof of Lemma \ref{Lemma: Suff McAfee} thus reveals that it suffices to check the SHRC on the intervals $[a,x)$ and $(x,b]$ with the appropriate left and right limits/derivatives up to $x=c$. We have
\begin{align*}
    \frac{\partial}{\partial x} \left( x + \frac{F(x)}{f(x)} \right) &= \begin{cases}
    \frac{3}{2} &, a\leq x <c,\\
    \frac{3}{2} + \frac{(b-a)(b-c)}{2(b-x)^2} &, c<x\leq b
    \end{cases}
    \quad \geq 0,\\
    \frac{\partial}{\partial x} \left( x - \frac{1-F(x)}{f(x)} \right) &= \begin{cases}
    \frac{3}{2} + \frac{(b-a)(c-a)}{2(x-a)^2} &, a\leq x <c,\\
    \frac{3}{2} &, c<x\leq b
    \end{cases}
    \quad \geq 0.
\end{align*}

\phantom{a}\\
{\bfseries Truncated standard normal distribution:}\\
We have $F(x) = \tfrac{\Phi(x)-\Phi(0)}{\Phi(1)-\Phi(0)}, f(x) = \tfrac{\phi(x)}{\Phi(1)-\Phi(0)}, f'(x) = -x f(x)$. We find
\begin{align*}
    \frac{\partial}{\partial x} \left( x + \frac{F(x)}{f(x)} \right) &= 2 +x \cdot  \frac{\Phi(x) - \Phi(0)}{\phi(x)} \geq 0,\\
    \frac{\partial}{\partial x} \left( x - \frac{1- f(x)}{F(x)} \right) &= 2 - x \cdot \frac{\Phi(1)-\Phi(x)}{\phi(x)}  \geq 2 - 1 \cdot \frac{\Phi(1)-\Phi(0)}{\phi(1)} > \frac{1}{2} > 0,
\end{align*}
for $x \in [0,1]$. Similar calculations can be made for other truncated normal distributions.

\phantom{a}\\
{\bfseries A class of beta distributions:}\\
We only consider the case $\beta = 1 \leq \alpha$ as the case $\alpha=1 \leq \beta$ follows similarly. Note that $F(x) = x^{\alpha}, f(x) = \alpha x^{\alpha - 1}, f'(x) = \alpha(\alpha-1) x^{\alpha - 2}$, $x \in [0,1]$. We have
\begin{align*}
    \frac{\partial}{\partial x} \left( x + \frac{F(x)}{f(x)} \right) &= 1+ \frac{1}{\alpha} \geq 0,\\
    \frac{\partial}{\partial x} \left( x - \frac{1-F(x)}{f(x)} \right) &= 2 + \underbrace{\tfrac{\alpha - 1}{\alpha}}_{\leq 1}  \underbrace{\left(x^{-\alpha} - 1\right)}_{\geq -1}  \geq 1 > 0
\end{align*}
for all $x \in (0,1]$.

\end{proof}

The proof  of Theorem \ref{Theorem:divider characterization} is accompanied by some upfront Lemmata for a piecewise twice continuously differentiable\footnote{Not necessarily continuous itself.} distribution function $G$ with all left and right limits and strictly quasiconcave function $\pi_G(p \mid x_D) = u(x_D - p)G(2p) + u(p)(1-G(2p))$ with $m_G(x_D) := m_{\pi_G(.\mid x_D)}$.

\begin{lemma}\label{Lemma: strict concavity in 2p=xD}
$\pi_{G}(p \mid x_D)$ is strictly concave in $2p=x_D$.
\begin{proof}
Let $G=\Gp, \Gm$ denote the right resp. left limit function and $g^{\pm}$,$g'^{\pm}$ be the left resp. right first resp. second derivative of $G$. Note that the first left/right derivative is
\begin{align}
    \frac{\partial^{\pm}}{\partial p} \pi_G(p \mid x_D) =& -u'(x_D-p) G^{\pm}(2p) + u'(p)(1-G^{\pm}(2p)) \nonumber\\
    & + 2 g^{\pm}(2p)(u(x_D-p)-u(p)).\label{eq: first derivative general}
\end{align}
The second left/right derivative is
\begin{align*}
    \frac{\partial^{2 \pm}}{\partial p^2} \pi_G(p \mid x_D) = &u''(x_D-p)G^{\pm}(2p) + u''(p)(1-G^{\pm}(2p))\\
    &- 4g^{\pm}(2p) \left( u'(x_D-p) + u'(p) \right)\\
    &+ 4 g'^{\pm}(2p)(u(x_D-p)-u(p)).
\end{align*}
Evaluating in $2p=x_D$ from left and right yields
\begin{equation*}
    \frac{\partial^{2 \pm}}{\partial p^2} \pi_G(p \mid x_D) \res_{2p=x_D} = u''(\tfrac{x_D}{2}) - 8 g^{\pm}(x_D) u'(\tfrac{x_D}{2}) < 0,
\end{equation*}
proving strict concavity there.
\end{proof}
\end{lemma}

\begin{lemma}\label{Lemma: median and optimal price}
We have
\begin{equation*}
    \begin{cases}
    x_D < 2m_G(x_D) &, G(x_D) < \tfrac{1}{2},\\
    x_D = 2m_G(x_D) &, G^{-}(x_D) \leq \tfrac{1}{2} \leq G(x_D),\\
    2m_G(x_D) < x_D &, \tfrac{1}{2} < G^-(x_D),
    \end{cases}
\end{equation*}
where $\Gp = G$ resp. $\Gm$ denote the right resp. left limit function of $G$. Furthermore, if $G(x_D) < \tfrac{1}{2}$, then $\Gm(2m_G(x_D)) < \tfrac{1}{2}$ and $\tfrac{1}{2} < G^-(x_D)$ implies $\tfrac{1}{2} < \Gm(2m_G(x_D))$.
\end{lemma}
\begin{proof}
See Lemma \ref{Lemma: strict concavity in 2p=xD} for the notation and recall
\begin{align*}
    \frac{\partial^{\pm}}{\partial p} \pi_G(p \mid x_D) =& -u'(x_D-p) G^{\pm}(2p) + u'(p)(1-G^{\pm}(2p))\\
    & + 2 g^{\pm}(2p)(u(x_D-p)-u(p)).
\end{align*}
Evaluating this in $2p=x_D$ (from left and right) yields
\begin{equation}\label{eq: derivatives in 2p=xD}
\frac{\partial^{\pm}}{\partial p} \pi_G \res_{2p = x_D} = u'(\frac{x_D}{2}) \cdot (1-2\Gpm(x_D)).
\end{equation}
Now if $\Gm(x_D) \leq G(x_D)<\tfrac{1}{2}$ than \eqref{eq: derivatives in 2p=xD} is positive and thus so is the first derivative in an open neighborhood of $2p=x_D$ by continuity of the resp. functions. Thus $x_D < 2m_G(x_D)$.
The additional statement follows as an adjustment of the end of the proof of Proposition \ref{Proposition: McAfee}: By strict quasiconcavity, we must have
\begin{align*}
    0 &\leq \frac{\partial^-}{\partial p} \pi_G(p \mid x_D) \res_{p = m_G(x_D)}\\
    &= -u'(x_D - m_G(x_D)) \Gm(2m_G(x_D)) + u'(m_G(x_D)) (1-\Gm(2m_G(x_D))\\
    &\phantom{==}+ 2\gm(2m_G(x_D))(u(x-m_G(x_D)) - u(m_G(x_D)))\\
    &< u'(m_G(x_D)) \cdot (1 - 2 \Gm(2m_G(x_D))),
\end{align*}
revealing $\Gm(2m_G(x_D)) < \tfrac{1}{2}$.

If $\tfrac{1}{2}< \Gm(x_D) \leq G(x_D)$, a similar argument reveals $2m_G(x_D) < x_D$ and $\Gm(2m_G(x_D)) < \tfrac{1}{2}$.

Finally, consider $\Gm (x_D) \leq \tfrac{1}{2} \leq G(x_D)$. If both inequalities are strict, we necessarily have $2m(x_D)=x_D$ by strict quasiconcavity. If $\Gm (x_D) = \tfrac{1}{2} = G(x_D)$, $2p=x_D$ must be an isolated local maximum by Lemma \ref{Lemma: strict concavity in 2p=xD} and thus $x_D = 2m_G(x_D)$ by strict quasiconcavity. Assume now exactly one inequality is strict (the other binding), e.g., $\Gm(x_D)=\tfrac{1}{2} < G(x_D)$ (the other case follows analogously). Then, $\tfrac{\partial^+}{\partial p}\pi_G \res_{2p=x_D}<0$ by equation \eqref{eq: derivatives in 2p=xD}. Furthermore, $\tfrac{\partial^-}{\partial p} \pi_G$ is positive on a neighborhood $2p \in (x_D-\delta, x_D)$ for some $\delta >0$ by Lemma \ref{Lemma: strict concavity in 2p=xD}. Thus, strict quasiconcavity reveals $2m_G(x_D) = x_D$.
\end{proof}

\begin{lemma} \label{Lemma: price announcement increasing}
The optimal price announcement $m_G(x_D)$ is increasing. If in addition $G$ is differentiable in $2m_G(x_D)$, $m_G(x_D)$ is strictly increasing in a neighborhood of $x_D$.
\begin{proof}
Recall from the proof of Lemma \ref{Lemma: strict concavity in 2p=xD} that
\begin{align*}
    \frac{\partial^{\pm}}{\partial p} \pi_G(p \mid x_D) =& -u'(x_D-p) \Gpm(2p) + u'(p)(1-\Gpm(2p))\\
    & + 2 \gpm(2p)(u(x_D-p)-u(p)).
\end{align*}
Consequently,
\begin{align}
    \frac{\partial}{\partial x_D} \frac{\partial^{\pm}}{\partial p} \pi_G(p \mid x_D) &= -u''(x_D-p) \Gpm(2p) + 4 \gpm(2p) u'(x_D-p) \geq 0.\label{eq: derivative p and x}
\end{align}
By strict quasiconcavity we know that
\begin{align*}
    \frac{\partial^{+}}{\partial p} \pi_G(p \mid x_D) \res_{p = m_G(x_D)} &\leq0\\
    \frac{\partial^{-}}{\partial p} \pi_G(p \mid x_D) \res_{p = m_G(x_D)} &\geq0 
\end{align*}
with a strict inequality each in a small neighborhood for $m_G(x_D)<p$ resp. $p<m_G(x_D)$. Thus, the inequality \eqref{eq: derivative p and x} reveals that $m(x_D)$ is increasing, though not necessarily strictly increasing. But it is strictly increasing if $\gm>0$ (making \eqref{eq: derivative p and x} strict) and $\tfrac{\partial^{-}}{\partial p} \pi_G(p \mid x_D) \res_{p = m_G(x_D)} = 0$ which are implied by, e.g., $G(2p)$ being differentiable in $p = m_G(x_D)$ with positive derivative. If $G$ is even continuously differentiable with positive derivative, applying the implicit function theorem reveals continuity of $m_G$.
\end{proof}
\end{lemma}

\begin{proof}[Proof of Theorem \ref{Theorem:divider characterization}]
Note that $g \hat{=} \pi_{G_0}$ and $f \hat{=}\pi_{G_1}$ fulfill Proposition \ref{Proposition: stict quasi concave minimum mh} from Appendix \ref{Appendix: Quasiconcavity}. Now, applying Lemma \ref{Lemma: median and optimal price} and realizing $G_i = \Gm_i$ by continuity\footnote{There is no conflict with the possible discontinuity of $G_0$ in $\xu$ as $2m_G(x_D) = \xu$ is only possible if $x_D = \xu$ and $\Gm(\xu) <1$, where the middle case is active.} yields the functional form of the optimal price announcement. Finally, Lemma \ref{Lemma: price announcement increasing} concludes the proof.
\end{proof}

\begin{proof}[Proof of Corollaries \ref{Corollary: B xm and fix points} and \ref{Corollary: B cont and str increasing}]
Most statements follow immediately from the proofs of Lemmata \ref{Lemma: median and optimal price} and \ref{Lemma: price announcement increasing} and since $x_D \mapsto \tfrac{x_D}{2}$ is continuous and strictly increasing. For the continuity of $m(x_D)$ in the 'glueing points' $x_D \in \{ \mu_{G_1}^-, \mu_{G_0}^+ \}$ for Corollary \ref{Corollary: B cont and str increasing}, we can apply a sandwich argument using Corollary \ref{Corollary: B xm and fix points} for the limits coming from outside $[\mu_{G_1}^-, \mu_{G_0}^+]$.
\end{proof}

\begin{proof}[Proof of Lemma \ref{Lemma: Worst-case distributions chooser}]
$\Phi_C(x_C) = \min_{G \in \cG} \mathbb{E}_{G} \left[ \max\{ u(x_C - m(z)), u(m(z)) \} \right]$.

If $x_C < 2m(\xl)$, the max-function will always choose $u(m(z))$. Thus, as $z \mapsto u(m(z))$ is increasing, the worst-case is attained by $G = G_1$ which is stochastically dominated by every other distribution in $\cG$ and puts most weight on the lowest price announcements. Interim worst-case EU is thus $\mathbb{E}_{G_1}[u(m(z))]$.
        
If $x_C > 2m(\xu)$, the max-function will always choose $u(x_C - m(z))$. Thus, as $z \mapsto u(x_C - m(z))$ is decreasing, the worst-case is attained by $G = G_0$, i.e. the distribution stochastically dominating all other distributions in $\cG$ and puts most weight on the highest price announcements. Interim worst-case EU is thus $\mathbb{E}_{G_0}[u(x_C - m(z))]$.

If $x_C \in [2m(\xl),2m(\xu)]$, the worst-case is putting as much weight at (and around) the valuation $z_*$ that induces a price announcement $2m(z_*) = x_C$ as possible. This value is unique since $m$ is strictly increasing by our assumptions and Corollary \ref{Corollary: B cont and str increasing}. More precisely, consider the partition $X = [\xl, z_*) \cup [z_*, \xu]$. Note that for any $x \in [\xl, z_*)$, the max-function will choose $u(x_C - m(z))$, thus the expectation is minimized by $G_0$. For $x \in [z_*, \xu]$ the max-function selects $u(m(z))$, thus the expectation is minimized by $G_1$. Piecing both cases together into a distribution function gets us $G^*_{x_C} \in \cG$ (it must be right-continuous, thus $G^*_{x_C}(z_*) = G_1(z_*)$). Note that interim worst-case EU in this case can be written as
    \begin{align*}
        &\mathbb{E}_{G^*_{x_C}}[\max\{ u(x_C - m(z)), u(m(z)) \} ]\\
        = &\int_{[\xl,z_*)} u(x_C - m(z)) \ G_0(\mathrm{d} z) + (G_1(z_*) - G_0(z_*)) \cdot u(x_C/2)\\
        & \ + \int_{(z_*,\xu]} u(m(z)) \ G_1(\mathrm{d}z).
    \end{align*}
\end{proof}

\begin{proof}[Proof of Lemma \ref{Lemma: derivatives of interim expected utility}]
Before we get into the calculations, note that by our smoothness assumptions, $m_0, m_1$ are strictly increasing in $x_D$ and continuously differentiable on the considered resp. intervals: Outside their medians, they are zeros of a continuously differentiable function $\tfrac{\partial}{\partial p} \pi_{G_i}(p \mid x_D)$ and thus continuously differentiable themselves by the implicit function theorem (see Lemma \ref{Lemma: median and optimal price}). This implies that also $m$ is continuous, strictly increasing and piecewise continuously differentiable (Lemma \ref{Lemma: median and optimal price} and Theorem \ref{Theorem:divider characterization}).

Thus, both $\Phi_D$ and $\Phi_C$ are continuous (indeed it will turn out that they are even differentiable except for $\Phi_C$ in $x = 2m(\xl)$), so it suffices to show differentiability on the corresponding open intervals. The derivatives of $\Phi_D$ for the center interval and of $\Phi_C$ for the outer ones are obvious. We now turn towards the remaining ones.

Firstly, for $G \in \{ G_0, G_1\}$ consider $\Phi_D(x) = \pi_G(m(x) \mid x)$ on the respective (open) interval. A calculation with $u = \id$ (or using the envelope theorem) reveals
\begin{align*}
    \frac{\partial}{\partial x} \Phi_D(x) = &\frac{\partial}{\partial x} \left( (x-m(x)) \cdot G(2m(x) + m(x) \cdot (1-G(2m(x))) \right)\\
    =& G(2m(x)) + \frac{\partial}{\partial x} m'(x) \cdot \underbrace{ \left(\frac{\partial}{\partial p} \pi_G(p\mid x)\right) \res_{p = m(x)}}_{= 0}\\
    =& G(2m(x)).
\end{align*}

Secondly, consider $\Phi_C$ for $x \in (2m(\xl), 2m(\xu))$ and recall $z_* = z_*(x) = m^{-1}(x/2)$. We now apply a measure theoretic version of the differentiation of parameter integrals for the following calculation. Therefore, $C := \xu + 1$ serves as a constant bounding $z_*$ from above, $\mu^G$ describes the probability measure associated to the distribution function $G \in \{G_0,G_1 \}$ with densities $g_0,g_1$, and we again use the integration variable $z$ to avoid confusion.
\begin{align*}
    &\frac{\partial}{\partial x} \Phi_C(x)\\
    =& \frac{\partial}{\partial x} \left( \int_{[\xl,z_*)} x - m(z) \ G_0(\mathrm{d} z) + (G_1(z_*) - G_0(z_*)) \cdot \frac{x}{2} + \int_{(z_*,\xu]} u(m(z)) \ G_1(\mathrm{d}z) \right)\\
    =& \ \phantom{+} \int_{[\xl,z_*)} 1 \, G_0(\mathrm{d}z) - (x - m(z_*)) \cdot \frac{\partial}{\partial x} \left( \mu^{G_0}((z_*, C))  \right)\\
    & \ + \left( g_1(z_*) \cdot \left( \frac{\partial}{\partial x} (z_*) \right) - g_0(z_*) \cdot \left( \frac{\partial}{\partial x} (z_*) \right) \right) \cdot \frac{x}{2} + (G_1(z_*) - G_0(z_*)) \cdot \frac{1}{2}\\
    & \ + \int_{(z_*,\xu]} 0 \, G_1(\mathrm{d}z) + m(z_*) \cdot \frac{\partial}{\partial x} \left( \mu^{G_1}([z_*,C]) \right).
\end{align*}
Realizing $\mu^{G_0}((z_*,C)) = 1-G_0(z_*)$ we find $\tfrac{\partial}{\partial x} \mu^{G_0}((z_*,C]) = - g_0(z_*) \cdot \frac{\partial}{\partial x} (z_*)$ and similarly for $G_1$. Since $m(z_*) = x/2$ we thus obtain
\begin{align*}
    &\frac{\partial}{\partial x} \Phi_C(x)\\
    =& G_0(z_*) -  \frac{x}{2} \cdot (-g_0(z_*) \frac{\partial}{\partial x}(z_*)) +  \left( g_1(z_*) \cdot \left( \frac{\partial}{\partial x} (z_*) \right) - g_0(z_*) \cdot \left( \frac{\partial}{\partial x} (z_*) \right) \right) \cdot \frac{x}{2}\\
    & \ + \frac{1}{2} \cdot (G_1(z_*) - G_0(z_*)) + \frac{x}{2} \cdot (-g_1(z_*) \frac{\partial}{\partial x}(z_*))\\
    =& \frac{1}{2} \cdot (G_1(z_*) + G_0(z_*)).
\end{align*}
Finally, since the derivatives are non-negative, $\Phi_D$ and $\Phi_C$ are increasing.
\end{proof}

\begin{proof}[Proof of Theorem \ref{Theorem: interim EU comparison}]\phantom{a}

{\bfseries Case $\a < \b$:}\\
Under this assumption, we have $G_1(z) - G_0(z) < 1$ for all $z \in X$. In other words, for all $z \in X$ there is no $G \in \cG$ that can assign point mass $1$ to $z$. Since $m$ is measurable, as it's strictly increasing, we find $\mu^G (\{z \, | \,m(z) = \tfrac{x}{2}\}) < 1$ for all $G \in \cG$.

We now prove that $\Phi_D(x) < \Phi_C(x)$ for all $x \in [\mu_{G_1}, \mu_{G_0}]$. On the one hand for any such $x$  we have $m(x) = \tfrac{x}{2}$ and thus $\Phi_D(x) = \tfrac{x}{2}$. On the other hand
\begin{align*}
    \Phi_C(x) = \min_{G \in \cG} \mathbb{E}_G \left[ \underbrace{\max \{ x - m(z), m(z)}_{\geq \tfrac{x}{2}} \}  \right] > \tfrac{x}{2} = \Phi_D(x),
\end{align*}
where the strict inequality results from the argument above as the max can only be $\tfrac{x}{2}$ if $2m(z) = x$.\smallskip

In the following, we compare the derivatives on the outer intervals. We start with $x \in [\xl, \mu_{G_1})$. We have the following chain of arguments
\begin{align*}
    x &<\mu_{G_1}\\
    \implies x &< 2m(2m(x)) \qquad (\text{Corollary } \ref{Corollary: B xm and fix points})\\
    \implies m^{-1}(\tfrac{x}{2}) & < 2m(x)\\
    \implies G_1(m^{-1}(\tfrac{x}{2})) &< G_1(2m(x)) \qquad (\star)\\
    \implies \frac{1}{2} \left( G_1(m^{-1}(\tfrac{x}{2})) + G_0(m^{-1}(\tfrac{x}{2}))  \right) &< G_1(2m(x)) \qquad (G_0 \leq G_1)\\
    \implies \Phi'_C(x) &< \Phi'_D(x),
\end{align*}
where we note that $\Phi'_C(x)$ might be zero, but cannot be $1$ since $x < 2m_1(x)=2m(x) < 2m(\xu)$ as $2m$ is strictly increasing. The inequality in $(\star)$ remains a strict one, since (using Corollary \ref{Corollary: B xm and fix points}) $m^{-1}(\tfrac{x}{2}) < x < \mu_{G_1} \leq \b$ and thus, by its functional form, $G_1$ is strictly increasing in a neighborhood of $m^{-1}(\tfrac{x}{2})$.\smallskip

We now turn towards the case $x \in (\mu_{G_0}, \xu]$. In analogy to the above, the following chain of arguments applies:
\begin{align*}
    \mu_{G_0} &< x\\
    \implies 2m(2m(x)) &< x\\
    \implies 2m(x) &< m^{-1}(\tfrac{x}{2})\\
    \implies G_0(2m(x)) &< G_0(m^{-1}(\tfrac{x}{2}))\\
    \implies G_0(2m(x)) &< \frac{1}{2} \left(G_0(m^{-1}(\tfrac{x}{2}) + G_1(m^{-1}(\tfrac{x}{2}) \right)\\
    \implies \Phi'_D(x) &< \Phi'_C(x),
\end{align*}
where we note that indeed $\Phi'_C(x)$ could be equal to $1$ and that applying $G_0$ preserves the strict inequality since $\a \leq \mu_{G_0} < x < m^{-1}(\tfrac{x}{2})$.\smallskip

Together, the above steps imply $\Phi_D(x) < \Phi_C(x)$ for all $x \in X$.\bigskip

{\bfseries Case $\b \leq \a$:}\\
For $x \in [\b,\a]$ we have by Theorem \ref{Theorem:divider characterization} $z_* := m^{-1}(\tfrac{x}{2}) = x$ and there is $G \in \cG$, namely $G = G^*_{x}$, with $\mu^G(\{ z_* | m(z_*) = \tfrac{x}{2} \}) = 1$. Hence, the argument of the first part of the previous case reveals that $\Phi_D(x)=\Phi_C(x) = \frac{x}{2}$ if $x = z_* \in [\b,\a]$ and $\Phi_D(x) < \Phi_C(x)$ otherwise.
\end{proof}

\begin{proof}[Explicit formulae for $\Phi_C$ in Example \ref{Example: interim wc EU}]
For $\tfrac{1}{2} < \epsilon \leq 1$ we find
\begin{align*}
    \Phi_C(x) = \begin{cases}
    \frac{1}{4} x^2 + \frac{1}{2}\epsilon x + \frac{1}{4} (1-\epsilon)^2 &, 0\leq x< 1-\epsilon,\\
    \frac{x}{2} &, 1-\epsilon \leq x \leq  \epsilon,\\
    \frac{1}{4} x^2 + \frac{1}{2}(1-\epsilon) x + \frac{1}{4}\epsilon^2 &, \epsilon < x \leq 1,
    \end{cases}
\end{align*}
while for $\tfrac{1}{4} < \epsilon \leq \tfrac{1}{2}$ it is given by
\begin{align*}
    \Phi_C(x) = \begin{cases}
     \frac{1}{8}\epsilon^2 - \frac{1}{2}\epsilon + \frac{9}{32}&, 0\leq x <2m_1(0) = \frac{1}{4}-\frac{1}{2}\epsilon,\\
     \frac{1}{2}x^2 + (\epsilon - \frac{1}{4})x + \frac{1}{2}\epsilon^2 - \frac{3}{4}\epsilon + \frac{5}{16} &, \frac{1}{4}-\frac{1}{2}\epsilon \leq x < \frac{1}{2} - \epsilon\\
    \frac{1}{4}x^2 + \frac{1}{2}\epsilon x + \frac{1}{4}\epsilon^2 - \frac{1}{2}\epsilon + \frac{1}{4} &, \frac{1}{2}-\epsilon \leq x < \epsilon,\\
    \frac{1}{2}x^2 + \frac{1}{2}\epsilon^2 - \frac{1}{2}\epsilon + \frac{1}{4} &, \epsilon \leq x \leq 1-\epsilon,\\
    \frac{1}{4}x^2 +\frac{1}{2}(1- \epsilon) x + \frac{1}{4}\epsilon^2  &, 1-\epsilon < x \leq \frac{1}{2}+\epsilon,\\
    \frac{1}{2}x^2 + (\frac{1}{4}-\epsilon)x + \frac{1}{2}\epsilon^2 + \frac{1}{4}\epsilon + \frac{1}{16} &, \frac{1}{2}+\epsilon < x \leq \frac{3}{4}+\frac{1}{2}\epsilon,\\
    x + \frac{1}{8}\epsilon^2 - \frac{1}{2}\epsilon - \frac{7}{32} &, \frac{3}{4}+\frac{1}{2}\epsilon = 2m_0(1) < x \leq 1,
    \end{cases}
\end{align*}
and for $0 \leq \epsilon \leq \tfrac{1}{4}$ we finally have
\begin{align*}
    \Phi_C(x) = \begin{cases}
    - \frac{3}{8} \epsilon^2 - \frac{1}{4}\epsilon + \frac{1}{4} &, 0 \leq x < 2m_1(0) = \frac{1}{4} - \frac{1}{2}\epsilon,\\
    \frac{1}{2} x^2 + (\epsilon - \frac{1}{4})x - \frac{1}{2}\epsilon + \frac{9}{32} &, \frac{1}{4}-\frac{1}{2}\epsilon \leq x < \frac{1}{4},\\
    x^2 + (\epsilon - \frac{1}{2})x - \frac{1}{2}\epsilon + \frac{5}{16} &, \frac{1}{4} \leq x < \frac{1}{2} - \epsilon,\\
    \frac{1}{2}x^2 - \frac{1}{2}\epsilon^2 + \frac{3}{16} &, \frac{1}{2}-\epsilon \leq x \leq \frac{1}{2} + \epsilon,\\
    x^2 - (\epsilon + \frac{1}{2})x + \frac{1}{2}\epsilon + \frac{5}{16} &, \frac{1}{2} + \epsilon < x \leq \frac{3}{4}\\
    \frac{1}{2}x^2 + (\frac{1}{4} - \epsilon)x + \frac{1}{2}\epsilon + \frac{1}{32} &, \frac{3}{4} < x \leq \frac{3}{4} + \frac{1}{2}\epsilon,\\
    x - \frac{3}{8}\epsilon^2 - \frac{1}{4} \epsilon - \frac{1}{4} &, \frac{3}{4} + \frac{1}{2}\epsilon = 2m_0(1) < x \leq 1.
    \end{cases}
\end{align*}
\end{proof}

\begin{proof}[Proof of Proposition \ref{Proposition: correlation triangular} and calculations for Example \ref{Example: correlation}]
Note that a triangular distribution $\Tri^c$ on $[0,1]$ with mode $c$ has the CDF
\begin{equation*}
    \Tri^c(x) = \begin{cases}
    \frac{x^2}{c} &, x < c,\\
    c &, x=c\\
    1-\frac{(1-x)^2}{(1-c)} & c < x.
    \end{cases}
\end{equation*}
In $\cG(\Tri^c,\epsilon)$ the medians of $G_0^c, G_1^c$ can be derived to be
\begin{align*}
    \mu_0^c &= \begin{cases}
    \sqrt{\frac{c}{2}} + \epsilon &, \frac{1}{2}\leq c,\\
    1+\epsilon - \sqrt{\frac{1-c}{2}} &, c \leq \frac{1}{2}
    \end{cases}\\
     \mu_1^c &= \begin{cases}
    \sqrt{\frac{c}{2}} - \epsilon &, \frac{1}{2}\leq c,\\
    1-\epsilon - \sqrt{\frac{1-c}{2}} &, c \leq \frac{1}{2}
    \end{cases}
\end{align*}

In the case of correlation we thus find by Theorem \ref{Theorem:divider characterization}
\begin{align*}
    \begin{cases}
    \sqrt{\frac{x_D}{2}}-\epsilon \leq x_D \leq \sqrt{\frac{x_D}{2}}+\epsilon &, \text{if } x_D=c \leq \frac{1}{2}\\
    1-\epsilon - \sqrt{\frac{1-x_D}{2}} \leq x_D \leq 1+\epsilon - \sqrt{\frac{1-x_D}{2}} &, \text{if } \frac{1}{2} \leq x_D=c
    \end{cases}.
\end{align*}
Note that for $x_D \leq \tfrac{1}{2}$ we always have $x_D \leq \sqrt{\tfrac{x_D}{2}}+\epsilon$ and for $\tfrac{1}{2} \leq x_D$ the statement $1-\epsilon - \sqrt{\tfrac{1-x_D}{2}} \leq x_D$ is always true.\\
In the case $c = \tfrac{1}{2}$ and $\epsilon = \tfrac{1}{5}$ we have $\mu_0 = \tfrac{7}{10}, \mu_1 = \tfrac{3}{10}$. The optimal price announcement outside of $[0.3,0.7]$ can again be derived by the corresponding FOCs and SOCs. Explicitly, they are given by
\begin{equation*}
    \begin{cases}
    \frac{x}{6} +\frac{1}{60} \cdot \sqrt{100x^2 + 40x + 79}-\frac{1}{15} &, x_D < \tfrac{3}{10},\\
     \frac{x}{6} - \frac{1}{60} \cdot  \sqrt{100x^2 - 240x + 219}+\frac{2}{5}&, \tfrac{7}{10} < x_D.
    \end{cases}
\end{equation*}

\end{proof}

\end{document}